\pdfoutput=1 
\documentclass[11pt,letterpaper]{article}
\newcommand{\ifarxiv}[2]{#2}

\usepackage[utf8]{inputenc}
\ifarxiv{
\usepackage[margin=0.9in]{geometry}
}{
\usepackage[margin=1in]{geometry}
}
\usepackage{palatino}
\usepackage{microtype}

\usepackage{helvet}
\usepackage{amsmath}
\usepackage{amsfonts}
\usepackage{amssymb}
\usepackage{amsthm}
\usepackage{comment}
\usepackage{mathtools}
\usepackage{mathrsfs}
\usepackage{bm}
\usepackage{caption}
\usepackage{subcaption}
\usepackage{sectsty}
\allsectionsfont{\sffamily}

\usepackage[dvipsnames]{xcolor}
\usepackage{tikz}
\usepackage{mathpazo}
\usepackage{verbatim}
\usepackage{tabularx}
\usepackage{tcolorbox}
\usepackage{float}

\usepackage[linesnumbered,ruled,lined]{algorithm2e}
\usepackage{algpseudocode}

\SetKwFor{ForPar}{for}{in parallel do}{end}
\SetKwFor{ForAllPar}{forall}{in parallel do}{end}

\usepackage{hyperref}
\hypersetup{colorlinks=true,citecolor=blue, linkcolor=blue, urlcolor=blue}

\usepackage[capitalize, nameinlink]{cleveref}

\usepackage{todonotes}
\def\lis{{\mathcal{L}}}
\def\xALG{{X^{\textbf{ALG}}}}
\def\tPA{{T^{\mathrm{PA}}}}
\def\tGD{{T^{\mathrm{GD}}}}
\def\tFD{{T^{\mathrm{FD}}}}
\def\dTV{{d_{\mathrm{TV}}}}
\def\xGD{{X_{\tGD}}}
\def\G{{\mathcal{G}}}
\def\p{{\mathbf{p}}}
\def\randseed{{\mathcal{R}}}
\def\GDResult{{Y}}
\def\R{{\mathfrak{R}}}
\def\pred{{\mathrm{pred}}}
\def\Pred{{\mathrm{Pred}}}
\def\init{{X'}}
\newcommand{\NC}{$\mathsf{NC}$}
\newcommand{\RNC}{$\mathsf{RNC}$}

\SetKwFor{ForAllPar}{forall}{in parallel do}{end}

\usepackage{xifthen}
\newcommand{\e}{\mathrm{e}}
\newcommand{\norm}[1]{\left\Vert#1\right\Vert}

\newcommand{\tuple}[1]{\left(#1\right)} 
\newcommand{\inner}[2]{\left\langle #1,#2\right\rangle}
\newcommand{\tp}{\tuple}
\newcommand{\abs}[1]{\left\vert#1\right\vert}
\newcommand{\ctp}[1]{\left\lceil#1\right\rceil}

\newcommand{\DTV}[2]{\dTV\left({#1},{#2}\right)}
\newcommand{\DKL}[2]{d_{\-{KL}}\left({#1},{#2}\right)}
\newcommand{\parGD}{\mathsf{parallel\text{-}RC}}
\newcommand{\ParGD}[4]{\parGD\left({#1},{#2},{#3},{#4}\right)}

\def\*#1{\boldsymbol{#1}} 
\def\+#1{\mathcal{#1}} 
\def\-#1{\mathrm{#1}} 
\def\^#1{\mathscr{#1}} 

\def\oPr{\*{\-{Pr}}}
\renewcommand{\Pr}[2][]{ \ifthenelse{\isempty{#1}}
  {\oPr\left[#2\right]}
  {\oPr_{#1}\left[#2\right]} } 

\def\oE{\mathbb{E}}
\newcommand{\E}[2][]{ \ifthenelse{\isempty{#1}}
  {\oE\left[#2\right]}
  {\oE_{#1}\left[#2\right]} }

\def\oVar{\mathrm{Var}}
\newcommand{\Var}[2][]{ \ifthenelse{\isempty{#1}}
  {\oVar\left[#2\right]}
  {\oVar_{#1}\left[#2\right]} }

\def\oEnt{\mathrm{Ent}}
\newcommand{\Ent}[2][]{ \ifthenelse{\isempty{#1}}
  {\oEnt\left[#2\right]}
  {\oEnt_{#1}\left[#2\right]} }

\newtheorem{theorem}{Theorem}[section]

\newtheorem*{claim*}{Claim}

\newtheorem{lemma}[theorem]{Lemma}
\newtheorem{proposition}[theorem]{Proposition}

\theoremstyle{definition}
\newtheorem{condition}{Condition}

\newtheorem{remark}[theorem]{Remark}
\newtheorem*{remark*}{Remark}

\newtheorem*{case*}{Case}

\title{Efficient Parallel Ising Samplers via Localization Schemes}
\date{}
\author{Xiaoyu Chen ~ Hongyang Liu ~ Yitong Yin ~ Xinyuan Zhang
}

\author{Xiaoyu Chen~\thanks{State Key Laboratory for Novel Software Technology, New Cornerstone Science Laboratory, Nanjing University, 163 Xianlin Avenue, Nanjing, Jiangsu Province, China. E-mails: \url{chenxiaoyu233@smail.nju.edu.cn}, \url{liuhongyang@smail.nju.edu.cn}, \url{yinyt@nju.edu.cn}, \url{zhangxy@smail.nju.edu.cn}} 
\and Hongyang Liu\footnotemark[1]
\and Yitong Yin\footnotemark[1]
\and Xinyuan Zhang\footnotemark[1]
}

\begin{document}
\maketitle
\begin{abstract}
We introduce efficient parallel algorithms for sampling from the Gibbs distribution and estimating the partition function of Ising models. 
These algorithms achieve parallel efficiency, with polylogarithmic depth and polynomial total work,
and are applicable to Ising models in the following regimes:
(1)~Ferromagnetic Ising models with external fields; 
(2)~Ising models with interaction matrix $J$ of operator norm $\|J\|_2<1$.

Our parallel Gibbs sampling approaches are based on localization schemes, 
which have proven highly effective in establishing rapid mixing of Gibbs sampling. 
In this work, we employ two such localization schemes to obtain efficient parallel Ising samplers: 
the \emph{field dynamics} induced by \emph{negative-field localization}, and \emph{restricted Gaussian dynamics} induced by \emph{stochastic localization}.
This shows that localization schemes are powerful tools, not only for achieving rapid mixing but also for the efficient parallelization of Gibbs sampling.
\end{abstract}

\setcounter{page}{0} 
\thispagestyle{empty}
\newpage
\section{Introduction}


The Ising model, originally introduced by Lenz and \cite{ising1924beitrag} in statistical physics to study the criticality of ferromagnetism, 
has since found numerous applications and has been extensively studied across various fields, 
including probability theory, statistical learning, and computer science.

Let $G = (V, E)$ be a connected undirected graph with $n=|V|$ vertices and $m=|E|$ edges. 
Let $\*\beta \in (0, +\infty)^E$ represent the edge activities, and $\*\lambda \in [0, 1]^V$ represent the external fields.
The \emph{Gibbs distribution} $\mu^{\-{Ising}}_{\*\beta, \*\lambda}$ of the {Ising model} on the graph $G$ with parameters $\*\beta$ and $\*\lambda$ is supported on $2^V$ and is given by:
\begin{align*}
    \forall S \subseteq V, \quad \mu^{\-{Ising}}_{\*\beta, \*\lambda}(S) := \frac{1}{Z^{\-{Ising}}_{\*\beta, \*\lambda}} \prod_{e \in m(S)} \beta_e \prod_{v \in S} \lambda_v,
\end{align*}
where $m(S) := \{e = (u, v) \in E \mid u, v \in S \text{ or } u, v \not\in S\}$ is the set of ``monochromatic'' edges,  
i.e., edges where both endpoints are either in $S$ or both outside of $S$. 
The normalizing factor, called the \emph{partition function}, is given by: $Z^{\-{Ising}}_{\*\beta, \*\lambda} := \sum_{S\subseteq E} \prod_{e \in m(S)} \beta_e \prod_{v \in S} \lambda_v$.

The Ising model is called \emph{ferromagnetic} if $\*\beta \in (1, +\infty)^E$, meaning that pairwise interactions favor monochromatic edges, 
and the model is called \emph{anti-ferromagnetic} if $\*\beta \in (0,1)^E$.

A central problem in the study of Ising models is sampling from the Gibbs distribution. This task is essential not only for estimating the partition function but also for various other inference and learning problems associated with the Ising model.

\subsection{Parallel sampler for ferromagnetic Ising model}

In a seminal work by \cite{jerrum1993polynomial}, it was shown that for ferromagnetic Ising models, approximately sampling from the Gibbs distribution and approximately computing the partition function are both tractable in polynomial time. 
This result, along with celebrated breakthroughs in polynomial-time approximations of permanents \cite{jerrum1989approximating,jerrum2004polynomial} and volumes of convex bodies \cite{dyer1991random}, 
remarkably showcased the power of random sampling in approximately solving \#\textsf{P}-hard inference problems on polynomial-time Turing machines.
 
%
%

%
Motivated by emerging applications in large-scale data and large models, parallel Gibbs sampling has recently drawn considerable attention~\cite{Liu2025Parallelize,AnariGR24,anari2023parallel,anari2023quadratic,liu2022simple,anari2021sampling,feng2019dynamic,feng2021distributed,fischer2018simple,feng2017sampling}.
Despite these advances, the tractability of the ferromagnetic Ising model through parallel algorithms remains largely unresolved, 
in contrast to the successes of classical polynomial-time sequential algorithms.

%
%
%

In the theory of computing, the parallel complexity of the ferromagnetic Ising model is of foundational significance.
In a seminal work, \cite{mulmuley1987matching} asked whether an \RNC{} algorithm (with polylogarithmic depth on polynomial processors) exists for sampling bipartite perfect matchings, which would imply an \RNC{} approximation algorithm for the permanent of Boolean matrices.
%
%
To this day, this remains a major open problem. \cite{teng1995independent} conjectured that no \RNC{} algorithm exists for this task, 
making it a rare example of a polynomial-time tractable problem suspected to be intrinsically sequential, yet not known to be \textsf{P}-complete.
The challenge of parallelizing the ferromagnetic Ising sampler is closely tied to the problem of sampling matchings, as both problems were resolved sequentially using the same canonical path argument~\cite{jerrum1989approximating,jerrum1993polynomial}.

In this work, we present an efficient parallel sampler in the \RNC{} class 
for the general ferromagnetic Ising model with nonzero external fields.
To the best of our knowledge, this is the first \RNC{} sampler for Ising model beyond the critical phase-transition threshold.

\begin{theorem}[ferromagnetic Ising sampler] 
\label{theorem:main-ising}
Let $\delta \in (0,1)$ be a constant.
There is a parallel algorithm that,
given $\epsilon \in (0,1)$ and an Ising model on a graph $G = (V,E)$ with parameters $\*\beta \in [1 + \delta, +\infty)^E$ and $\*\lambda \in [0, 1 - \delta]^V$,
outputs a sample from the Gibbs distribution $\mu^{-{Ising}}_{\*\beta, \*\lambda}$ within total variation distance $\epsilon$
in $(\epsilon^{-\frac{1}{\log n}} \log n)^{O_\delta(1)}$ parallel time using $O_\delta(m^2 \log\frac{n}{\epsilon})$ processors.
\end{theorem}

%


One can view this result as a parallel counterpart to \cite{jerrum1993polynomial}.
Through a standard reduction via non-adaptive annealing, 
this \RNC{} Ising sampler can be turned into an \RNC{} algorithm for approximating the Ising partition function $Z^{\-{Ising}}_{\*\beta, \*\lambda}$. 
Specifically, applying the work-efficient parallel annealing algorithm recently developed in \cite{liu2024work}, 
the parallel Ising sampler stated in \Cref{theorem:main-ising} can be transformed to a randomized parallel algorithm which returns an $(1\pm\epsilon)$-approximation of the Ising partition function in $(\epsilon^{-\frac{1}{\log n}} \cdot \log n)^{O_\delta(1)}$ parallel time on $\tilde{O}_\delta(m^2n/\epsilon^2)$ machines.

The parallel Ising sampler in \Cref{theorem:main-ising} is implied by a parallel sampler for the \emph{random cluster model} (see \Cref{sec:def-random-cluster} for  definition).
The correspondence between the Ising model and the random cluster model is due to the well-known \emph{Edwards-Sokal coupling}~\cite{edwards1988gen} (see \Cref{lemma:ES-coupling}).

\begin{theorem}[random cluster sampler]
\label{theorem:main-random-cluster}
Let $\delta\in(0,1)$ be a constant. 
There is a parallel algorithm that, 
given $\epsilon \in (0, 1)$ and a random cluster model on a graph $G=(V,E)$ with parameters $\*p \in [\delta, 1)^E$ and $\*\lambda \in [0, 1 - \delta]^V$,
outputs a sample from the random cluster distribution~$\mu^{\-{RC}}_{E,\*p, \*\lambda}$ within total variation distance $\epsilon$ 
in $(\epsilon^{-\frac{1}{\log n}} \cdot \log n)^{O_\delta(1)}$ parallel time on $O_\delta(m^2 \log\frac{n}{\epsilon})$ processors.
\end{theorem}


\Cref{theorem:main-ising} follows from \Cref{theorem:main-random-cluster} via the Edwards-Sokal coupling, as formally described in \Cref{sec:our-algo}.
%
Our main technical contribution in this part is a parallel sampler for the random cluster model, as stated in \Cref{theorem:main-random-cluster},
whose efficiency is proved in \Cref{sec:coupling-stationary}.

\subsection{Parallel sampler for Ising models with contracting interaction matrix}

An Ising model can be formulated through the interaction matrix.
Let $V$ be a set of $n$ vertices.
The \emph{Gibbs distribution} $\mu^{\-{Ising}}_{J, h}$ of the {Ising model}, supported on $\{\pm 1\}^V$, is defined by the \emph{interaction matrix} $J \in \mathbb{R}^{V \times V}$, which is a symmetric positive semi-definite matrix, and the \emph{external fields} $h \in \mathbb{R}^V$:
%
%
%
\begin{align} \label{eq:def-Ising-J}
\forall \sigma \in \{\pm 1\}^V, \quad \mu^{\-{Ising}}_{J, h} \propto \exp\tp{\frac{1}{2} \sigma^\top J\sigma + h^\intercal \sigma}.
\end{align}



Recent studies have established rapid mixing for Ising models in terms of their interaction norm~\cite{eldan2021spectral,anari2021entropic}, showing that Glauber dynamics mixes rapidly when $\|J\|_2<1$, whereas for $\|J\|_2>1$ it mixes slowly~\cite{LLP10}.
For parallel sampling, \cite{liu2022simple} proposed an \RNC{} algorithm under a stricter condition that requires $J$ to have a constant $\ell_\infty$-norm.
However, under the broader condition $\|J\|_2<1$, the best known parallel algorithms \cite{lee2023parallelising,Liu2025Parallelize} achieve only a depth of $\widetilde{O}(\sqrt{n})$, highlighting a significant gap in achieving optimal parallelism.

In this work, we present an \RNC{} parallel sampler for the Ising model under the condition $\|J\|_2<1$. 
To the best of our knowledge, this is the first \RNC{} Ising sampler that matches this critical threshold for rapid mixing in terms of the interaction norm.

\begin{theorem}[Ising sampler]
\label{theorem:gaussian-dynamics}
Let $\eta \in (0,1)$ be a constant.
There is a parallel algorithm that,
given $\epsilon \in (0, 1)$ and an Ising model with an interaction matrix $J \in \mathbb{R}^{V \times V}$ satisfying 
$$0 \prec \frac{\eta}{2} I \preceq J \preceq (1-\frac{\eta}{2}) I$$
and arbitrary external fields $h \in \mathbb{R}^V$, 
outputs a sample from the Gibbs distribution $\mu^{\-{Ising}}_{J, h}$ within total variation distance $\epsilon$ 
in $O_\eta(\log^4\tp{\frac{n}{\epsilon}})$ parallel time using $\widetilde{O}_\eta(n^3/\epsilon^2)$ processors. 

\end{theorem}

\begin{remark}
Note that for the Ising model with interaction matrix $J$, for every constant $C$, taking $J' = J + C I$ results in the same Gibbs distribution.
Thus, in the regime where $\|J\|_2\leq 1-\eta$, 
for which Glauber dynamics is known to mix rapidly \cite{eldan2021spectral,anari2021entropic}, 
one can add a diagonal matrix to the interaction matrix $J$ in order to satisfy the condition of \cref{theorem:gaussian-dynamics}.
\end{remark}

The \RNC{} sampler stated in \cref{theorem:gaussian-dynamics} is presented in \cref{sec:gaussian-dynamics}.

\subsection{Technique overview}

Classic local Markov chains, such as Glauber dynamics, are inherently sequential, updating the spin of a single site at each step based on its neighbors.
Recently, \cite{liu2022simple,Liu2025Parallelize} proposed a generic parallelization framework using \emph{correlated sampling} (or \emph{universal coupling}), enabling polylogarithmic-depth simulation of Glauber dynamics under a relaxed Dobrushin condition.
When applied to the Ising model, this yields an \RNC{} sampler in the uniqueness regime $\*\beta \in \left(\frac{\Delta-2}{\Delta}, \frac{\Delta}{\Delta-2} \right)^E$, where $\Delta$ is the maximum degree.
However, these conditions do not hold for the Ising models considered in this work.


Instead of directly parallelizing a local Markov chain, our approach focuses on simulating ``global'' Markov chains that update $O(n)$ spins in each step, satisfying the following:
\begin{itemize}
    \item The chain mixes in polylogarithmic steps;
    \item Each step of the chain has an efficient parallel implementation in \RNC{}.
\end{itemize}
These together ensure an \RNC{} sampler with  polylog depth and polynomial total work.

For the ferromagnetic Ising model, this global chain corresponds to the \emph{field dynamics} introduced by \cite{chen2021rapid}, while for the Ising model with a contracting interaction matrix, the global chain is the \emph{restricted Gaussian dynamics} (also known as the \emph{proximal sampler}) introduced by \cite{lee2021structured}.
Both of these global Markov chains can be viewed as arising from different localization schemes for the Ising Gibbs sampling.

Proposed by~\cite{chen2022localization}, 
the localization schemes provide an abstract framework that generalizes high-dimensional expander walks and encompasses a wide range of stochastic processes.
This framework can be interpreted in terms of noising and denoising processes, as discussed in \cite{EM22, chen2024rapid}.
Let $X \sim \mu$ be drawn from the target distribution $\mu$. 
A noisy channel $D$, which is a Markov chain, is applied to $X$ to obtain a ``noised'' sample $Y \sim D(X, \cdot)$. 
The joint distribution of $(X,Y)$ is denoted as $\pi$. 
The denoising process $U$ represents the time reversal of $D$, where it draws $Z$ from the posterior distribution $U(Y,\cdot) := \pi(\cdot \mid Y)$.
%
The Markov chains associated with the localization scheme update the current state $X$ to a new state $Z$ according to the rules:
\begin{itemize}
    \item Add noise to $X$ through the noisy channel: sample $Y \sim D(X, \cdot)$;
    \item Denoise $Y$ via the posterior distribution: sample $Z \sim U(Y, \cdot)$.
\end{itemize}
In particular, when the noisy channel $D$ corresponds to the continuous-time down walk channel, 
which drops each element in a set $X$ with a fixed probability, 
the localization process is called the \emph{negative-field localization}, 
and the above associated global Markov chain is the \emph{field dynamics}. 
Alternatively, when $D$ corresponds to a Gaussian noise channel, 
which adds Gaussian noise to the sample $X$, 
the localization process is called \emph{stochastic localization}, 
and the resulting global Markov chain becomes the \emph{restricted Gaussian dynamics}.

The localization scheme has previously proven remarkably effective in establishing rapid mixing of Glauber dynamics up to criticality ~\cite{chen2021rapid, chen2022optimal, anari2022entropic,  chen2022localization, chen2023nearlinear, anari2024trickle, chen2024rapid}.
Specifically, the localization process can effectively ``tilt'' the model parameters
towards sub-critical directions.
This allows the mixing properties established in sub-critical regimes to be conserved up to criticality, 
provided the associated global chains mix rapidly.

In the current work, we explore another perspective on localization schemes: 
leveraging the associated global Markov chains to achieve parallel efficiency in sampling.
This task is highly non-trivial, as these chains were originally designed as analytical tools for studying mixing times, 
and their efficient parallel simulation poses significant challenges.


For the ferromagnetic Ising model with external fields, 
we leverage the field dynamics of the random-cluster model to design a parallel sampler.
The field dynamics is the global Markov chain induced by negative-field localization.
Each update consists of two steps: 
\begin{enumerate}
    \item A noising step $Y \sim D(X, \cdot)$ which drops each element in $X$ with a fixed probability (as in \Cref{line:field-dynamics-denoising} of \Cref{alg:field}), and is straightforward to implement in parallel.
    \item A denoising step $Z \sim U(Y, \cdot)$, which requires sampling from a tilted posterior distribution  $U(Y,\cdot)$ and is non-trivial.
    We simulate the denoising step via Glauber dynamics on $U(Y,\cdot)$, 
which corresponds to a sub-critical (low-temperature) random-cluster model,
where rapid mixing is easier to establish.
\end{enumerate}
A similar idea was employed in \cite{chen2023nearlinear} to develop a near-linear time sequential sampler.
Here, we further parallelize this Glauber dynamics on  $U(Y,\cdot)$ using the parallel simulation algorithm introduced in~\cite{Liu2025Parallelize}, which yields \Cref{alg:BP-simulation}.
A key challenge, however, is that despite being tilted to the low-temperature regime,
the Dobrushin condition required in~\cite{Liu2025Parallelize} to ensure efficient parallel simulation does not hold.
To overcome this, we establish a new ``\emph{coupling with the stationary}'' criterion, 
which significantly relaxes the Dobrushin condition and guarantees the efficiency of parallel simulation of Glauber dynamics.
Altogether, this yields parallel random-cluster and ferromagnetic-Ising samplers with polylogarithmic depth and polynomial total work.

For the Ising model with a contracting interaction matrix $J$, where $\|J\|_2<1$, 
we leverage the restricted Gaussian dynamics (also known as the proximal sampler) to design a parallel sampler.
This dynamics defines a global Markov chain, which is the associated chain of the stochastic localization for the Ising model.
Each update consists of two steps: 
\begin{enumerate}
    \item A noising step $Y \sim D(X, \cdot)$, which introduces Gaussian noise to the current sample.
    \item A denoising step $Z \sim U(Y, \cdot)$, which resamples from a posterior distribution $U(Y, \cdot)$.
\end{enumerate}
With appropriately chosen parameters, $U(Y,\cdot)$ becomes a product distribution, 
allowing an efficient parallel implementation of the denoising process.
The non-trivial step is the noising step $Y\sim D(X,\cdot)$, 
which samples from a high-dimensional Gaussian distribution with mean $X$ and covariance $J^{-1}$ (as in \Cref{line:gaussian-noising} of \Cref{alg:gaussian}).
A key observation is that the Gaussian distribution is log-concave, so the noising step can be efficiently approximated
using recently developed parallelizations of Langevin Monte Carlo~\cite{anari2024fast}.

\section{Preliminaries}
\subsection{Notation}
Given a finite nonempty  ground set $E$, we will use boldface letters, such as $\*p$, to denote a vector in $\mathbb{R}^E$.
Given a vector $\*a \in \mathbb{R}^E$ and function $f: \mathbb{R} \to \mathbb{R}$, 
we write $\*b = f(\*a)$ for the vector $\*b \in \mathbb{R}^E$ satisfying that $b_e=f(a_e)$ for all $e\in E$.

Throughout the paper, we use $\log$ to denote the natural logarithm.

\subsection{Random cluster model} \label{sec:def-random-cluster}
Let $G=(V,E)$ be a connected, undirected graph.
The random cluster model on $G=(V,E)$ with edge probabilities $\*p \in [0,1]^E$ and vertex weights $\*\lambda \in [0,1]^V$ is defined as follows.
The random cluster distribution $\mu^{\mathrm{RC}}_{E,\*p,\*\lambda}$ is supported on the power set $2^E$, 
where each $S \subseteq E$ is assigned a weight:
\begin{align}
\label{eq:RC-weight}
    w^{\mathrm{RC}}_{E,\*p,\*\lambda} (S) := \prod_{e\in S} p_e \prod_{e\in E\setminus S} (1 - p_e) \prod_{C \in \kappa(V, S)} \tp{1 + \prod_{j \in C} \lambda_j},
\end{align}
where $\kappa(V,S)$ is the set of connected components in the graph $(V,S)$.
The distribution $\mu^{\mathrm{RC}}_{E,\*p,\*\lambda}$ is then given by:
\begin{align*}
\forall S \subseteq E, \quad \mu^{\mathrm{RC}}_{E,\*p,\*\lambda} (S) := \frac{w^{\mathrm{RC}}_{E,\*p,\*\lambda} (S)}{Z^{\mathrm{RC}}_{E,\*p,\*\lambda}},
\end{align*}
where the normalizing factor $Z^{\mathrm{RC}}_{E,\*p,\*\lambda} := \sum_{S \subseteq E} w^{\mathrm{RC}}_{E,\*p,\*\lambda} (S)$ is the {partition function}.


\subsection{Markov chains and mixing time}
Let $\Omega$ be a finite state space, and let $(X_t)_{t \ge 0}$ be a Markov chain over $\Omega$ with transition matrix $P$. 
We use $P$ to represent the Markov chain for convenience.
It is well known that a finite Markov chain $P$ converges to a unique stationary distribution $\mu$ if $P$ is irreducible and aperiodic (see~\cite{levin2017markov} for the definitions of these concepts).
The \emph{mixing time} of a Markov chain $P$ is defined by
\begin{align*}
    T_{\mathrm{mix}}(\epsilon) := \max_{x_0 \in \Omega} \min \{t \in \mathbb{N} \mid d_{\mathrm{TV}} (P^t(x_0,\cdot),\mu) < \epsilon \},
\end{align*}
where $d_{\mathrm{TV}}(\*p,\*q):=\frac{1}{2}\|\*p-\*q\|_1$ is the total variation distance.

\subsubsection{Glauber dynamics}\label{sec:Glauber-dynamics}
The Glauber dynamics is a canonical Markov chain for sampling from joint distributions.
Let $\mu$ be a distribution over $2^E$ on a finite ground set $E$. 
For any $X\subseteq E$ and $u\in E$, let $\mu_u^{X}$ be the marginal probability of $u$ given $X$, formally defined as:
\begin{align*}
\mu_u^{X}:=\frac{\mu(X \cup \{u\})}{\mu(X \cup \{u\}) + \mu(X \setminus \{u\})}.
\end{align*}
The Glauber dynamics $P^{\mathrm{GD}}$ updates a configuration $X \subseteq E$ according to the following rule:
\begin{itemize}
\item Pick an element $u \in E$ uniformly at random;
\item With probability $\mu_u^{X}$, replace $X$ with $X \cup \{u\}$; otherwise, replace $X$ with $X \setminus \{u\}$.
\end{itemize}
It is well known that $\mu$ is the unique stationary distribution of this chain~\cite{levin2017markov}.

\subsubsection{Field dynamics}\label{sec:field-dynamics}
Field dynamics is a novel Markov chain introduced in \cite{chen2021rapid}. 
Let $\mu$ be a distribution over $2^E$ on a finite ground set $E$. The field dynamics $P^{\mathrm{FD}}_\theta$ with parameter $\theta \in (0,1)$ updates a configuration $X \subseteq E$ according to the following rule:
\begin{itemize}
\item Add each element $u \in E$ into a set $S$ independently with probability $\theta$; 
\item Replace $X$ with a random $Y\subseteq E$ that follows the law of $\theta^{-1} * \mu$ conditioned on $Y \subseteq X \cup S$,
where $\theta^{-1}*\mu$ denotes the distribution  supported on $2^E$ and defined as:
\begin{align*}
\forall T \subseteq E,\quad (\theta^{-1} *\mu)(T) \propto \theta^{-\abs{T}} \mu(T).
\end{align*}
\end{itemize}
The field dynamics can be thought of as an adaptive block dynamics, where a block of sites is chosen for updating, adapted to the current configuration. It was shown in \cite{chen2021rapid} that $\mu$ is the unique stationary distribution of the field dynamics.

For the random cluster model on graph $G=(V,E)$ with parameters $\*p$ and $\*\lambda$, 
the second step of the field dynamics corresponds to:
\begin{itemize}
\item Replace $X$ with a random subgraph $Y$ distributed as $\mu^{\mathrm{RC}}_{S \cup X,\*p^\star,\lambda}$, where $\p^\star = \frac{\*p}{\*p+\theta(1-\*p)}$.
\end{itemize}

\subsection{Model of computation}
We assume the concurrent-read exclusive-write ($\mathsf{CRCW}$) parallel random access machine ($\mathsf{PRAM}$) \cite{jeje1992introduction} as our model of computation,
where concurrent writes to the same memory location are allowed, and an arbitrary value written concurrently is stored.
%
The computational complexity is measured by the number of rounds (parallel time steps) and the number of processors (or machines) used in the worst case.

We use \NC{} to refer to both the class of algorithms that run in polylogarithmic time using a polynomial number of processors and the class of problems solvable by such algorithms. \RNC{} denotes the randomized counterpart of \NC{}.


\section{Parallel random-cluster and Ising samplers via field dynamics} \label{sec:our-algo}








We propose a parallel algorithm for sampling from the random-cluster model, based on the simulation of field dynamics.
Consider a random-cluster model defined on a graph  $G=(V,E)$ with edge probabilities  $\*p \in (0,1]^E$ and vertex weights $\*\lambda \in [0,1)^V$.

The main algorithm is described in \Cref{alg:field}.

\begin{algorithm}[H] 
\SetKwInOut{Input}{Input} 
\SetKwInOut{Parameter}{Parameter}
\SetKwInOut{Output}{Output} 
\Input{Graph $G(V,E)$, $\p\in (0,1)^E, \*\lambda\in [0,1)^V$, and error bound $\epsilon\in (0,1)$.}
\Parameter{$\theta\in(0,1)$, $N_0 >0$ and integer $\tFD\ge 1$.}
\Output{A random configuration $X\subseteq E$ satisfying $d_{\text{TV}}\left(X,\mu^{\text{RC}}_{E,\p,\*\lambda}\right) \le \epsilon $.}
initialize $X \gets E$\;
 \lIf{$\abs{V} \leq N_0$}{
    \Return{$X\sim\mu^{\-{RC}}_{E, \*p, \*\lambda}$} by brute force
 }
  \ForAll{$i=1,2,\ldots,T^{\mathrm{FD}}$}{
      construct $S\subseteq E$ by including each $e \in E$ independently with probability $\theta$\; \label{line:field-dynamics-denoising}
      $X\gets\ParGD{(V,S \cup X)}{\p^\star}{ \*\lambda}{\left(2T^{\mathrm{FD}}\right)^{-1}\epsilon}$, where $\*p^\star = \frac{\*p}{\*p+\theta (1-\*p)}$\;
  }

\Return{X}\;
\caption{Random cluster field dynamics}\label{alg:field}
\end{algorithm}

\Cref{alg:field} simulates the \emph{field dynamics} (defined in \Cref{sec:field-dynamics}), the Markov chain induced by the \emph{negative-field localization process} for the random cluster model $\mu^{\text{RC}}_{E,\p,\*\lambda}$.
At each transition step, the algorithm samples from $\mu^{\mathrm{RC}}_{S \cup X,\p^\star,\lambda}$, where $\p^\star = \frac{\*p}{\*p+\theta(1-\*p)}$. This step is implemented using the $\parGD$ subroutine, described in \Cref{alg:BP-simulation}.

The subroutine $\ParGD{G}{\p}{\*\lambda}{\epsilon}$, as presented in \Cref{alg:BP-simulation}, 
is a parallel algorithm designed for sampling from the random cluster distribution $\mu^{\text{RC}}_{E, \p, \*\lambda}$.
%
It simulates the Glauber dynamics for $\mu = \mu^{\text{RC}}_{E, \p, \*\lambda}$ through a parallel simulation approach introduced in~\cite{liu2022simple}, using \emph{correlated sampling}, also known as \emph{universal coupling}. 
Specifically, \Cref{alg:BP-simulation} employs inverse transform sampling as the universal coupler.
However, prior analyses in~\cite{liu2022simple,Liu2025Parallelize} rely on a variant of the Dobrushin condition, which does not hold in our context, and thus overcoming this barrier require novel techniques.

%
%
%

\begin{algorithm}[H] 
\SetKwInOut{Input}{Input} 
\SetKwInOut{Parameter}{Parameter}
\SetKwInOut{Output}{Output} 
\Input{Graph $G(V,E)$,  $\p\in (0,1)^E,\*\lambda\in [0,1)^V$, and error bound $\epsilon\in (0,1)$.}
\Parameter{Integers $\tGD,\tPA\ge 1$.}
\Output{A random configuration $X\subseteq E$ satisfying $d_{\text{TV}}\left(X,\mu^{\text{RC}}_{E,\p,\*\lambda}\right) \le \epsilon $.}
\label{line:ini1}




generate $\init\subseteq E$ by including each $e\in E$ independently with probability $\mu^E_e$\; 
\label{line:ini-X}



\ForAllPar{$i=1,2,...,\tGD$}{
generate $e_i\in E$ and $\randseed_i\in [0,1]$ uniformly at random\;\label{line:random-choices}

$\GDResult^0_i\gets I[e_i\in \init]$\;
} 
\label{line:ini2}

\ForAllPar{$i=1,2,...,\tGD$ and $e\in E$}{
    
    calculate $\pred_i(e) \gets \max\{j\ge 1\mid j<i \land e_j=e \}\cup\{0\}$\;
    \label{line:pred}
} 




\For{$t=1,2,...,\tPA$}{ \label{line:iter}

\ForAllPar{$i=1,2,...,\tGD$}{

construct $\sigma = \sigma^t_i \subseteq E$ as: \hspace{300pt}
\mbox{\hspace{15pt}} $\sigma^t_i\gets\{e\in E\mid (\pred_i(e) = 0 \land e\in \init) \lor (\pred_i(e)\geq 1 \land \GDResult^{t-1}_{\pred_i(e)}=1)\}$\;  \label{line:sigma}

calculate marginal probability $\mu_{e_i}^\sigma\gets \frac{\mu(\sigma\cup\{e_i\})}{\mu(\sigma\cup\{e_i\})+\mu(\sigma\setminus \{e_i\})}$\;\label{line:weight}
\leIf{$\randseed_i<\mu_{e_i}^\sigma$}
{$\GDResult^t_{i}\gets1$}{$Y^t_{i}\gets0$} \label{line:update-rules}
}
}


$X \gets \left\{ e\in E \mid (j_e = 0 \land e\in \init) \lor (j_e\geq 1 \land \GDResult^{\tPA}_{j_e}=1), 
j_e\gets \max\{j\ge 1\mid e_j=e \}\cup\{0\} \right\}$\; \label{line:res}


\Return{X}\;
\caption{$\ParGD{G}{\p}{\*\lambda}{\epsilon}$}\label{alg:BP-simulation}
\end{algorithm}

Let $n=|V|$ and $m=|E|$. 
The parameters $\theta$, $N_0$, $\tFD$, $\tGD$, and $\tPA$ assumed by the algorithms are specified in \Cref{tabel:parameters}.
We will show that the algorithm is both efficient and correct under these parameters.

\renewcommand{\arraystretch}{1.43}
\begin{table}[ht]
    \centering
    \begin{tabularx}{\textwidth}{Xcp{0.6\textwidth}}
      \hline
      & parameter & value \\
      \hline
      & $\theta$ & $\e^{-100} {\color{black} \exp\tp{\frac{10\log(\epsilon/2)}{\log n}} } \exp\tp{-140(1 - \lambda_{\max})^{-2}}\cdot{p_{\min}}/{\log n}$ \\
      & $N_0$ & $\max \left\{\exp\tp{60\;(1-\lambda_{\max})^{-2}},\frac{3}{p_{\min}}, {\color{black} \sqrt{\frac{\log \tp{2/\epsilon^2}}{\log n}} } \right\}$ \\
      & $\tFD$ & $\left\lceil \tp{\frac{\e}{\theta}}^{5(1-\lambda_{\max})^{-2}}\tp{2 \log n + \log \log \frac{2}{p_{\min}} + \log \frac{2}{\epsilon^2}} \right\rceil$ \\
      & $\tGD$ & $\left\lceil 2m(\log m + \log(8\tFD/\epsilon)) \right\rceil$ \\
      & $\tPA$ & $\ctp{3 \log \tp{\frac{8\tGD\tFD}{\epsilon}}}$ \\ 
      \hline
    \end{tabularx}
    \caption{Parameters assumed by \Cref{alg:field}}
    \label{tabel:parameters}
\end{table}
\renewcommand{\arraystretch}{1}

Specifically, we show that if the random cluster model is in the low-temperature regime (i.e., corresponding to an Ising model with large $\beta>1$), 
then it holds simultaneously:
\begin{enumerate}
    \item 
The Glauber dynamics gets sufficiently close to $\mu^{\text{RC}}_{E,\p,\*\lambda}$ within $\tGD$ steps.
    \item 
The parallel iterative updates (the \textbf{for} loop in \Cref{line:iter} of \Cref{alg:BP-simulation}) stabilize in $\tPA$ rounds,
faithfully simulating the $\tGD$-step Glauber dynamics with high probability.
\end{enumerate}
Therefore, \Cref{alg:BP-simulation} is an efficient parallel sampler in the low-temperature regime.

As a result of the localization scheme,
the field dynamics in \Cref{alg:field} ``tilts'' the current instance $\mu^{\text{RC}}_{E,\p,\*\lambda}$  to a new random cluster distribution $\mu^{\text{RC}}_{{S\cup X},\p^\star,\*\lambda}$ at low temperature in each round. 
Within this regime, the Glauber dynamics is rapidly mixing and can be faithfully simulated by \Cref{alg:BP-simulation} in parallel.
The field dynamics mixes in polylogarithmic rounds according to \cite{chen2023nearlinear}.





Combining everything together, we obtain the following theorem for parallel sampler.

\begin{theorem} \label{thm:field-sim}
Let $\delta\in(0,1)$ be a constant,
and let $\*p \in [\delta, 1)^E$ and $\*\lambda \in [0, 1 - \delta]^V$. Then:
\begin{enumerate}
    \item \label{item:field-efficiency} \Cref{alg:field} halts in $(\epsilon^{-\frac{1}{\log n}} \cdot \log n)^{O_{\delta}(1)}$ parallel time using $O_{\delta}(m^2\log\frac{n}{\epsilon})$ processors.
    \item \label{item:field-correctness} \Cref{alg:field} outputs a random $X\subseteq E$ with $\DTV{X}{\mu^{\-{RC}}_{E, \*p, \*\lambda}} \leq \epsilon$.
\end{enumerate}
\end{theorem}

Our main theorem for the random cluster sampler (\Cref{theorem:main-random-cluster}) follows directly from \Cref{thm:field-sim}.
Our main theorem for the Ising sampler (\Cref{theorem:main-ising}) follows from \Cref{thm:field-sim} as well, through the Edwards-Sokal coupling~\cite{edwards1988gen}, which connects the distributions of the Ising and random cluster models. 
%
%
For our purposes, we use a variant of the ES coupling from \cite{feng2022sampling}.
%
%
%




\begin{lemma}[\text{\cite[Proposition 2.3]{feng2022sampling}}] \label{lemma:ES-coupling}
For the Ising model on a graph $G = (V,E)$ with parameters $\*\beta \in (1,+\infty)$ and $\*\lambda \in (0,1)$,
the following process generate a sample $X \sim \mu^{\mathrm{Ising}}_{\*\beta,\*\lambda}$:
\begin{enumerate}
\item Sample $S \subseteq E$ according to the distribution $\mu^{\mathrm{RC}}_{E,\*p,\*\lambda}$, where $\*p = 1 - \*\beta^{-1}$;

\item For each connected component $C$ in the graph $(V,S)$, add $C$ to $X$ with probability $\frac{\prod_{u \in C} \lambda_u}{1+\prod_{u \in C} \lambda_u}$. \label{item:step-rc-ising}
\end{enumerate}
\end{lemma}

Furthermore, the ES coupling can be efficiently implemented in parallel.

\begin{lemma} \label{lemma:time-step2-ES}
%
\Cref{item:step-rc-ising} of \Cref{lemma:ES-coupling} can be computed in $O(\log n)$ parallel time using $O(m)$ processors.
\end{lemma}


\begin{proof}
Use the connected component algorithm in \cite{shiloach1980log}, to compute all components of the graph $G' = (V, S)$ in $O(\log n)$ parallel time using $O(n + m)$ processors. For each connected component $C$, the product $\prod_{u \in C} \lambda_u$ can be computed in $O(\log |C|)$ parallel time using $O(|C|)$ processors by divide-and-conquer.  The overall complexity is $O(\log n)$ parallel time using $O(m)$ processors.
\end{proof}


\begin{proof}[Proof of \Cref{theorem:main-ising}]
    By \Cref{lemma:ES-coupling}, a sample from  $\mu^{\-{Ising}}_{\*\beta, \*\lambda}$  can be generated as follows:
    \begin{enumerate}
        \item Generate a random subset $S\subseteq E$ using \Cref{alg:field}, where $\*p = 1 - \*\beta^{-1}$;
        \item Apply \Cref{item:step-rc-ising} of \Cref{lemma:ES-coupling} to construct $X$.
    \end{enumerate}
    By \Cref{thm:field-sim}, we have $\DTV{S}{\mu^{\-{RC}}_{E, \*p, \*\lambda}} \leq \epsilon$,
    which, combined with \Cref{lemma:ES-coupling}, ensures that $\DTV{X}{\mu^{\-{Ising}}_{\*\beta, \*\lambda}} \leq \epsilon$.
    The efficiency follows from \Cref{thm:field-sim} and \Cref{lemma:time-step2-ES}. 
\end{proof}

It remains to prove \Cref{thm:field-sim}. The key steps of the proof are outlined in \Cref{sec:coupling-stationary}, while the full proof is deferred to \Cref{appendix:wrap-up}.


\section{A coupling-with-stationary criterion for parallel simulation}\label{sec:coupling-stationary}
In this section, we outline the key ideas underlying the proof of \Cref{thm:field-sim}.

Consider the random cluster model on $G = (V, E)$ with parameters $\*p \in (0, 1]^E$ and $\*\lambda \in [0, 1)^V$.
 The distribution $\mu=\mu^{\mathrm{RC}}_{E,\*p,\*\lambda}$ is defined over subgraphs $S\subseteq E$ and is proportional to the weight function $w=w^{\mathrm{RC}}_{E,\*p,\*\lambda}$, given by:
\begin{align*}
    \mu (S) \propto w(S):=\prod_{e\in S} p_e \prod_{f\in E\setminus S} (1 - p_e) \prod_{C \in \kappa(V, S)} \tp{1 + \prod_{j \in C} \lambda_j},
\end{align*}
where  $\kappa(V, S)$ denotes the set of all connected components in the graph $(V, S)$.

Consider the Glauber dynamics $(X_t)_{t\ge 0}$ for sampling from $\mu$, starting from an initial state $X_0 \subseteq E$.
At each step, the process updates $X_{i-1}$ to $X_i$ as follows: Select an edge $e_i\in E$ and a real number 
$\randseed_i\in [0,1]$ uniformly at random, and update $X_{i}$ according to:
\begin{align}
    \label{eq:GD-process}
    X_i = \begin{cases} 
    X_{i-1} \cup \{e_i\} & \text{if }\randseed_i < \mu_{e_i}^{X_{i-1}}\\
    X_{i-1} \setminus \{e_i\} & \text{otherwise}
    \end{cases},
\end{align}
where $\mu_{e_i}^{X_{i-1}}$ represents the marginal probability of $e_i$ given $X_{i-1}$, formally defined as
$$\mu_{e}^{X}:=\frac{\mu(X \cup\{e\})}{\mu(X \cup\{e\})+\mu(X \setminus \{e\})}=\frac{w(X \cup\{e\})}{w(X \cup\{e\})+w(X \setminus \{e\})}.$$ 
It is straightforward to verify that  $(X_t)_{t\ge 0}$ is the Glauber dynamics defined in \Cref{sec:Glauber-dynamics}.

%

We will show that \Cref{alg:BP-simulation} faithfully simulates the Glauber dynamics, provided that a “coupling with stationarity” condition holds, as formally stated below.

\begin{condition}
\label{cond:good-event-property}
Let $\eta \in (0,1)$ be a parameter.
There is a $\G \subseteq 2^E$ satisfying the following conditions:
\begin{enumerate}
    \item 
    \label{item:good-event-1}
    For any $X,Y \in \G$, the total variation in marginal probabilities across edges satisfies
    $$\sum_{e\in E} \left|\mu_e^X-\mu_e^Y\right| \leq \frac{|X \bigoplus Y |}{2}.$$
    \item
    \label{item:good-event-2}
    %
    %
    %
    For any $S: E\to 2^E$, consider the random configuration $\sigma\subseteq E$ generated by including each edge $e\in E$ independently with probability $\mu_e^{S(e)}$.
    Then, we have
    $$\Pr{\sigma \notin \G} \leq \eta.  
$$
%
\end{enumerate}
\end{condition}


\Cref{cond:good-event-property} describes a ``coupling with stationary'' style criterion, akin to the one introduced in \cite{hayes2006coupling}.
This criterion significantly relaxes the Dobrushin condition used in \cite{liu2022simple, Liu2025Parallelize} for ensuring the efficiency of parallel simulation, such as \Cref{alg:BP-simulation}.
Rather than requiring distance decay between all pairs of configurations, it only enforces distance decay in expectation between good configurations $X,Y\in\G$  within the one-step optimal coupling of Glauber dynamics,
while ensuring that these good configurations appear frequently throughout the dynamics.

We now state the main theorem of this section.
\begin{theorem} \label{thm:dtv-xALG-GD}
    If \Cref{cond:good-event-property} holds with parameter $\eta$, then for any $\tGD,\tPA\ge 1$, 
    the output $\xALG$ of \cref{alg:BP-simulation} and the Glauber dynamics $\left(X_t\right)_{0\le t\le\tGD}$ defined in \eqref{eq:GD-process}, with initial state $X_0\subseteq E$ generated as in \Cref{line:ini-X} of \cref{alg:BP-simulation}, 
    satisfy the following bound on the total variation distance: 
    \begin{align*}
    \DTV{\xALG}{X_{\tGD}} \leq 
    \tGD \cdot \left(2^{-(\tPA-1)}+4\eta \right).
    \end{align*}
\end{theorem}

\begin{remark}[generality of the criterion]
Note that
\Cref{alg:BP-simulation}, \Cref{cond:good-event-property}, and \Cref{thm:dtv-xALG-GD} are stated abstractly for Glauber dynamics applied to a general joint distribution $\mu$ over variables with a Boolean domain. 
Therefore, these results are applicable to any such distribution.
%
\end{remark}


    %

The following theorem is an application of \Cref{thm:dtv-xALG-GD} to low-temperature random cluster models, establishing that \Cref{alg:BP-simulation} is both accurate and efficient in this regime.

\begin{theorem} \label{thm:GD-correctness}
    Let $G = (V, E)$, $\*p, \*\lambda$, and $\epsilon$ be the input to \Cref{alg:BP-simulation}.
    Assume that $n=|V|\ge 3$, $\epsilon \leq\frac{1}{2}$, and the following condition holds:
    \begin{align*}
        (1 - p_{\min}) \log n &\leq \min\left\{\-{e}^{-40} \exp\tp{-\frac{5\log (1/\epsilon)}{\log n}}, \frac{1 - \lambda_{\max}}{27} \right\}.
    \end{align*}
    Then, for any $\tGD \ge \ctp{2m(\log m + \log(4/\epsilon))}$ and $\tPA \ge \lceil 3\log \left( 4\tGD/\epsilon \right) \rceil$,
    the output $\xALG$ of \Cref{alg:BP-simulation} satisfies
    \begin{align}
        \dTV\left(\xALG,\mu\right) \leq \epsilon.
    \end{align}
\end{theorem}

\Cref{thm:GD-correctness} is a key step in establishing the accuracy and efficiency of \Cref{alg:field}. 
Under the parameterization given in \Cref{tabel:parameters}, each instance of the random cluster model passed to \Cref{alg:BP-simulation} within \Cref{alg:field} falls within the low-temperature regime required by \Cref{thm:GD-correctness}.
The theorem holds because these low-temperature random cluster models satisfy \Cref{cond:good-event-property}.

The proofs of \Cref{thm:dtv-xALG-GD} and \Cref{thm:GD-correctness} are deferred to \Cref{appendix:dtv-xALG-GD} and \Cref{appendix:verify-cond-GD-sim}, respectively. Finally, the proof of \Cref{thm:field-sim} is completed in \Cref{appendix:wrap-up}.

\section{Parallel Ising sampler via restricted Gaussian dynamics}
\label{sec:gaussian-dynamics}
We introduce a parallel sampler for the Ising model with contracting interaction norm.
Our algorithm approximately implements the following \emph{restricted Gaussian dynamics} (also known as the \emph{proximal sampler}, described in \Cref{alg:gaussian}), 
the Markov chain induced by the \emph{stochastic localization process} for Ising Gibbs sampling.

\begin{algorithm}[H] 
\SetKwInOut{Input}{Input} 
\SetKwInOut{Output}{Output} 
\Input{Interaction matrix $J\in\mathbb{R}^{V\times V}$, external fields $h \in \mathbb{R}^V$, and $\epsilon\in (0,1)$.}
\Output{A random configuration $X\subseteq E$ satisfying $d_{\text{TV}}\left(X,\mu^{\-{Ising}}_{J, h}\right) \le \epsilon $.}
initialize $x_0 \in \{\pm 1\}^V$ arbitrarily\;
  \ForAll{$i=1,2,\ldots,T^{\mathrm{RGD}}$}{
      draw $y_i\sim \mathcal{N}(x_{i-1},J^{-1})$\label{line:gaussian-noising}\;
      draw $x_i = x \in \{\pm 1\}^V$ with prob. $\propto \mu^{\-{Ising}}_{J, h}(x) \cdot \exp\left(-\frac{1}{2}(y_i-x)^\top J(y_i-x)\right)$\label{line:gaussian-denoising}\;
  }

$X \gets x_{T^{\mathrm{RGD}}}$\;
\Return{X}\;
\caption{Restricted Gaussian Dynamics}\label{alg:gaussian}
\end{algorithm}

The efficient implementation of \Cref{alg:gaussian} follows from the following observations:
\begin{enumerate}
    \item Suppose $0 \prec \frac{\eta}{2} I \preceq J \preceq (1-\frac{\eta}{2}) I$. 
    Given an error bound $0<\epsilon_0<1$, 
    the Gaussian sampling step (\cref{line:gaussian-noising}) of \Cref{alg:gaussian} can be approximated within total variation error $\epsilon_0$.
    Specifically, there exists an algorithm that, given $x_{i-1}$ and $J^{-1}$, produces samples from a distribution $\pi$ satisfying $d_{\text{TV}}({\pi},{\mathcal{N}(x_{i-1},J^{-1})})\leq \epsilon_0$. This approximation can be achieved in $O_\eta(\log^3 (n/\epsilon_0^2))$ parallel time using $\widetilde{O}(n^3/\epsilon_0^2)$ processors.
    \label{item:gaussian-1}
    \item The denoising step (\cref{line:gaussian-denoising}) of \Cref{alg:gaussian} involves sampling from a product distribution whose marginals can be computed in {$O(\log n)$} parallel time. This step can be implemented faithfully with no error using {$O(\log n)$} parallel time.
    \label{item:gaussian-2}
    \item Assuming perfect simulations of \cref{line:gaussian-noising} and \cref{line:gaussian-denoising} in \Cref{alg:gaussian}, the total variation distance between the output $X$ and the target distribution $\mu^{\-{Ising}}_{J, h}$ does not exceed $\epsilon/2$ after {$T^{\mathrm{RGD}} = O_\eta(\log (n/\epsilon))$} iterations of the outer loop in~\Cref{alg:gaussian}.
    \label{item:gaussian-3}
\end{enumerate}

The detailed analysis of \cref{item:gaussian-1}, \cref{item:gaussian-2}, and \cref{item:gaussian-3} is provided in \cref{appendix:gaussian-dynamics}.

Assuming these properties hold, we now proceed with the proof of \cref{theorem:gaussian-dynamics}.
\begin{proof}[Proof of \cref{theorem:gaussian-dynamics}]
Let $X$ be the distribution of output generated by \cref{alg:gaussian}.
Under the assumption that the entire algorithm can be perfectly simulated, \cref{item:gaussian-3} ensures that the total variation distance between $X$ and the target distribution  $\mu^{\-{Ising}}_{J, h}$ is at most $\epsilon/2$.
Now, we implement \Cref{line:gaussian-noising} of \Cref{alg:gaussian} using an approximate oracle provided by \cref{item:gaussian-1}, setting the parameter $\epsilon_0 =\epsilon/2T^{\mathrm{RGD}}$.
Let $Y$ denote the output produced by our implementation of~\Cref{alg:gaussian}.
By a simple coupling argument, we obtain $d_{\text{TV}}(X,Y)\leq T^{\mathrm{RGD}}\cdot \epsilon_0 = \epsilon/2$.
Applying the triangle inequality, we have $d_{\text{TV}}(Y, \mu^{\-{Ising}}_{J, h})\leq d_{\text{TV}}(X,\mu^{\-{Ising}}_{J, h})+d_{\text{TV}}(X,Y)\leq \epsilon$.
Thus, our implementation produces an approximate sample within $O_\eta(\log^4\tp{n/\epsilon})$ parallel time using $\widetilde{O}_\eta(n^3/\epsilon^2)$ processors, as guaranteed by \cref{item:gaussian-1}, \cref{item:gaussian-2}, and \cref{item:gaussian-3}.
\end{proof}

\bibliographystyle{alpha}
\bibliography{references}

\newcommand{\etalchar}[1]{$^{#1}$}
\begin{thebibliography}{ABTV23}

\bibitem[ABTV23]{anari2023quadratic}
Nima Anari, Callum Burgess, Kevin Tian, and Thuy-Duong Vuong.
\newblock Quadratic speedups in parallel sampling from determinantal distributions.
\newblock In {\em Proceedings of the 35th ACM Symposium on Parallelism in Algorithms and Architectures}, pages 367--377, 2023.

\bibitem[ACV24]{anari2024fast}
Nima Anari, Sinho Chewi, and Thuy{-}Duong Vuong.
\newblock Fast parallel sampling under isoperimetry.
\newblock In Shipra Agrawal and Aaron Roth, editors, {\em The Thirty Seventh Annual Conference on Learning Theory, June 30 - July 3, 2023, Edmonton, Canada}, volume 247 of {\em Proceedings of Machine Learning Research}, pages 161--185. {PMLR}, 2024.

\bibitem[AGR24]{AnariGR24}
Nima Anari, Ruiquan Gao, and Aviad Rubinstein.
\newblock Parallel sampling via counting.
\newblock In Bojan Mohar, Igor Shinkar, and Ryan O'Donnell, editors, {\em Proceedings of the 56th Annual {ACM} Symposium on Theory of Computing, {STOC} 2024, Vancouver, BC, Canada, June 24-28, 2024}, pages 537--548. {ACM}, 2024.

\bibitem[AHL{\etalchar{+}}23]{anari2023parallel}
Nima Anari, Yizhi Huang, Tianyu Liu, Thuy-Duong Vuong, Brian Xu, and Katherine Yu.
\newblock Parallel discrete sampling via continuous walks.
\newblock In {\em Proceedings of the 55th Annual ACM Symposium on Theory of Computing}, pages 103--116, 2023.

\bibitem[AHSS21]{anari2021sampling}
Nima Anari, Nathan Hu, Amin Saberi, and Aaron Schild.
\newblock Sampling arborescences in parallel.
\newblock In {\em {ITCS}}, volume 185 of {\em LIPIcs}, pages 83:1--83:18. Schloss Dagstuhl - Leibniz-Zentrum f{\"{u}}r Informatik, 2021.

\bibitem[AJK{\etalchar{+}}21]{anari2021entropic}
Nima Anari, Vishesh Jain, Frederic Koehler, Huy~Tuan Pham, and Thuy-Duong Vuong.
\newblock Entropic independence i: Modified log-sobolev inequalities for fractionally log-concave distributions and high-temperature ising models.
\newblock {\em arXiv preprint arXiv:2106.04105}, 2021.

\bibitem[AJK{\etalchar{+}}22]{anari2022entropic}
Nima Anari, Vishesh Jain, Frederic Koehler, Huy~Tuan Pham, and Thuy{-}Duong Vuong.
\newblock Entropic independence: optimal mixing of down-up random walks.
\newblock In {\em {STOC}}, pages 1418--1430. {ACM}, 2022.

\bibitem[AKV24]{anari2024trickle}
Nima Anari, Frederic Koehler, and Thuy-Duong Vuong.
\newblock Trickle-down in localization schemes and applications.
\newblock In {\em Proceedings of the 56th Annual ACM Symposium on Theory of Computing}, pages 1094--1105, 2024.

\bibitem[CCYZ24]{chen2024rapid}
Xiaoyu Chen, Zongchen Chen, Yitong Yin, and Xinyuan Zhang.
\newblock Rapid mixing at the uniqueness threshold.
\newblock {\em arXiv preprint arXiv:2411.03413}, 2024.

\bibitem[CE22]{chen2022localization}
Yuansi Chen and Ronen Eldan.
\newblock Localization schemes: {A} framework for proving mixing bounds for markov chains (extended abstract).
\newblock In {\em {FOCS}}, pages 110--122. {IEEE}, 2022.

\bibitem[CFYZ21]{chen2021rapid}
Xiaoyu Chen, Weiming Feng, Yitong Yin, and Xinyuan Zhang.
\newblock Rapid mixing of {G}lauber dynamics via spectral independence for all degrees.
\newblock In {\em FOCS}, pages 137--148, 2021.

\bibitem[CFYZ22]{chen2022optimal}
Xiaoyu Chen, Weiming Feng, Yitong Yin, and Xinyuan Zhang.
\newblock Optimal mixing for two-state anti-ferromagnetic spin systems.
\newblock In {\em {FOCS}}, pages 588--599. {IEEE}, 2022.

\bibitem[Col88]{cole1988parallel}
Richard Cole.
\newblock Parallel merge sort.
\newblock {\em SIAM J. Comput.}, 17(4):770--785, 1988.

\bibitem[CZ23]{chen2023nearlinear}
Xiaoyu Chen and Xinyuan Zhang.
\newblock A near-linear time sampler for the ising model with external field.
\newblock In {\em SODA}, pages 4478--4503. SIAM, 2023.

\bibitem[DFK91]{dyer1991random}
Martin Dyer, Alan Frieze, and Ravi Kannan.
\newblock A random polynomial-time algorithm for approximating the volume of convex bodies.
\newblock {\em Journal of the ACM (JACM)}, 38(1):1--17, 1991.

\bibitem[EAM22]{EM22}
Ahmed El~Alaoui and Andrea Montanari.
\newblock An information-theoretic view of stochastic localization.
\newblock {\em IEEE Trans. Inform. Theory}, 68(11):7423--7426, 2022.

\bibitem[EKZ21]{eldan2021spectral}
Ronen Eldan, Frederic Koehler, and Ofer Zeitouni.
\newblock A spectral condition for spectral gap: fast mixing in high-temperature ising models.
\newblock {\em Probability Theory and Related Fields}, pages 1--17, 2021.

\bibitem[ES88]{edwards1988gen}
Robert~G. Edwards and Alan~D. Sokal.
\newblock Generalization of the {F}ortuin-{K}asteleyn-{S}wendsen-{W}ang representation and {M}onte {C}arlo algorithm.
\newblock {\em Phys. Rev. D (3)}, 38(6):2009--2012, 1988.

\bibitem[FG18]{fischer2018simple}
Manuela Fischer and Mohsen Ghaffari.
\newblock A simple parallel and distributed sampling technique: Local glauber dynamics.
\newblock In {\em 32nd International Symposium on Distributed Computing (DISC 2018)}, volume 121, pages 26--1. Schloss Dagstuhl-Leibniz-Zentrum f{\"u}r Informatik, 2018.

\bibitem[FGW22]{feng2022sampling}
Weiming Feng, Heng Guo, and Jiaheng Wang.
\newblock Sampling from the ferromagnetic ising model with external fields.
\newblock {\em arXiv preprint arXiv:2205.01985 [v1]}, 2022.

\bibitem[FHY21]{feng2021distributed}
Weiming Feng, Thomas~P Hayes, and Yitong Yin.
\newblock Distributed metropolis sampler with optimal parallelism.
\newblock In {\em SODA}, pages 2121--2140. SIAM, 2021.

\bibitem[FSY17]{feng2017sampling}
Weiming Feng, Yuxin Sun, and Yitong Yin.
\newblock What can be sampled locally?
\newblock In {\em Proceedings of the 36th {ACM} Symposium on Principles of Distributed Computing ({PODC})}, pages 121--130, 2017.

\bibitem[FVY19]{feng2019dynamic}
Weiming Feng, Nisheeth~K Vishnoi, and Yitong Yin.
\newblock Dynamic sampling from graphical models.
\newblock In {\em Proceedings of the 51st Annual ACM SIGACT Symposium on Theory of Computing (STOC)}, pages 1070--1081, 2019.

\bibitem[Hol23]{lee2023parallelising}
Lee Holden.
\newblock Parallelising glauber dynamics.
\newblock {\em RANDOM 2024}, 2023.

\bibitem[HV06]{hayes2006coupling}
Thomas~P. Hayes and Eric Vigoda.
\newblock Coupling with the stationary distribution and improved sampling for colorings and independent sets.
\newblock {\em Ann. Appl. Probab.}, 16(3):1297--1318, 2006.

\bibitem[Isi24]{ising1924beitrag}
Ernst Ising.
\newblock {\em Beitrag zur theorie des ferro-und paramagnetismus}.
\newblock PhD thesis, Grefe \& Tiedemann Hamburg, 1924.

\bibitem[J{\'e}J92]{jeje1992introduction}
Joseph J{\'e}J{\'e}.
\newblock An introduction to parallel algorithms.
\newblock {\em Reading, MA: Addison-Wesley}, 10:133889, 1992.

\bibitem[JS89]{jerrum1989approximating}
Mark Jerrum and Alistair Sinclair.
\newblock Approximating the permanent.
\newblock {\em SIAM J. Comput.}, 18(6):1149--1178, 1989.

\bibitem[JS93]{jerrum1993polynomial}
Mark Jerrum and Alistair Sinclair.
\newblock Polynomial-time approximation algorithms for the {I}sing model.
\newblock {\em SIAM J. Comput.}, 22(5):1087--1116, 1993.

\bibitem[JSV04]{jerrum2004polynomial}
Mark Jerrum, Alistair Sinclair, and Eric Vigoda.
\newblock A polynomial-time approximation algorithm for the permanent of a matrix with nonnegative entries.
\newblock {\em Journal of the ACM (JACM)}, 51(4):671--697, 2004.

\bibitem[LLP10]{LLP10}
David~A. Levin, Malwina~J. Luczak, and Yuval Peres.
\newblock Glauber dynamics for the mean-field {I}sing model: cut-off, critical power law, and metastability.
\newblock {\em Probab. Theory Related Fields}, 146(1-2):223--265, 2010.

\bibitem[LPW17]{levin2017markov}
David~A. Levin, Yuval Peres, and Elizabeth~L. Wilmer.
\newblock {\em Markov chains and mixing times}.
\newblock American Mathematical Society, Providence, RI, 2017.

\bibitem[LST21]{lee2021structured}
Yin~Tat Lee, Ruoqi Shen, and Kevin Tian.
\newblock Structured logconcave sampling with a restricted gaussian oracle.
\newblock In {\em The Annual Conference on Learning Theory (COLT)}, pages 2993--3050, 2021.

\bibitem[LY22]{liu2022simple}
Hongyang Liu and Yitong Yin.
\newblock Simple parallel algorithms for single-site dynamics.
\newblock In {\em STOC}, pages 1431--1444. ACM, 2022.

\bibitem[LY25]{Liu2025Parallelize}
Hongyang Liu and Yitong Yin.
\newblock Parallelize single-site dynamics up to dobrushin criterion.
\newblock {\em J. ACM}, 72(1), January 2025.

\bibitem[LYZ24]{liu2024work}
Hongyang Liu, Yitong Yin, and Yiyao Zhang.
\newblock Work-efficient parallel counting via sampling.
\newblock {\em arXiv preprint arXiv:2408.09719}, 2024.

\bibitem[MVV87]{mulmuley1987matching}
Ketan Mulmuley, Umesh~V Vazirani, and Vijay~V Vazirani.
\newblock Matching is as easy as matrix inversion.
\newblock In {\em STOC}, pages 345--354. ACM, 1987.

\bibitem[SV80]{shiloach1980log}
Yossi Shiloach and Uzi Vishkin.
\newblock An o (log n) parallel connectivity algorithm.
\newblock Technical report, Computer Science Department, Technion, 1980.

\bibitem[Ten95]{teng1995independent}
Shang-Hua Teng.
\newblock Independent sets versus perfect matchings.
\newblock {\em Theoret. Comput. Sci.}, 145(1-2):381--390, 1995.

\end{thebibliography}
\appendix
\section{Parallel simulation of Glauber dynamics (Proof of \texorpdfstring{\Cref{thm:dtv-xALG-GD}}{Theorem 4.1})} \label{appendix:dtv-xALG-GD}
 We assume that \Cref{cond:good-event-property} holds for a given parameter $\eta \in [0,1]$. 
 
 Let $Y^\ell= (Y^\ell_i)_{1\leq i\leq \tGD}$, where $Y^\ell_i$ is constructed in \Cref{line:iter} of \Cref{alg:BP-simulation}, 
 which indicates the result of the $i$-th update in the $\ell$-th iteration.
 Note that $Y^\ell$ for $1 \le \ell \le \tPA$ are determined by the random choices $\init$ and $\left(e_i,\randseed_i\right)_{i=1}^{\tGD}$ generated in \Cref{line:ini-X} and \Cref{line:random-choices} of \Cref{alg:BP-simulation}, respectively. 
The next lemma shows that $Y^\ell$ stabilizes fast with high probability.





\begin{lemma}
\label{lemma:GD-dependency-path} 
Given parameters $\tPA,\tGD \geq 1$, for any $1\leq \ell \leq \tPA$, it holds that 
\begin{align}
    \Pr{ Y^\ell \neq Y^{\ell-1} }
    \leq 
    \left( 2^{-(\ell-1)}  + 4 \eta  \right) \cdot \tGD.
\end{align}
\end{lemma}

\begin{proof}
For $1\leq i \leq \tGD$, define
$$\Pred_i := \left\{\pred_i(e):e\in E \right\} \setminus \{0\}.$$
Intuitively, $\Pred_i$ is the set of updates that have impact on the  result of the $i$-th update.

We claim that for $\ell \geq 2$ and $1\leq i \leq \tGD$,
\begin{align}
\label{eq:GD-induction-ineq}
\Pr{ Y^\ell_{i} \neq Y^{\ell-1}_{i} } \leq \left( \frac{1}{2m} \sum_{j\in \Pred_i}\Pr{ Y^{\ell-1}_{j} \neq Y^{\ell-2}_j } \right) + 2\eta.
\end{align}
Assuming~\eqref{eq:GD-induction-ineq}, the following can be proved by an induction on $1\le \ell\le\tPA$: 
\begin{align}
\label{eq:GD-one-site}
\Pr{ Y^\ell_{i} \neq Y^{\ell-1}_{i} } \leq 2^{-(\ell-1)}  + 4 \eta.
\end{align}
 For the induction basis $\ell=1$,~\eqref{eq:GD-one-site} holds trivially since $2^{-(\ell-1)}  + 4 \eta>1$. 
 We now assume~\eqref{eq:GD-one-site} holds for $\ell-1 \ge 1$. By~\eqref{eq:GD-induction-ineq} and the induction hypothesis,
\begin{align}
\Pr{ Y^\ell_{i} \neq Y^{\ell-1}_{i} } &\leq \frac{|\Pred_i|}{2m} \left( 2^{-(\ell-2)}  + 4 \eta \right) + 2 \eta
\nonumber
\\
\mbox{(Since $|\Pred_i|\leq m$)} \quad\, & \leq 2^{-(\ell-1)}  + 4 \eta.
\nonumber
\end{align}

\cref{lemma:GD-dependency-path} follows from~\eqref{eq:GD-one-site} by taking the union bound over all $1\leq i\leq \tGD$.
%
%

It remains to prove the claim~\eqref{eq:GD-induction-ineq}, which completes the proof of \cref{lemma:GD-dependency-path}.

Let $\R_{<i}$ denote the random choices used  by the first $i-1$ updates, formally 
\begin{align*}
\R_{<i} := \left(e_j,\randseed_j\right)_{1\leq j<i}.
\end{align*}
Recall the rules for $Y_i$ being updated in \Cref{line:update-rules}.
The probability of $Y_i^\ell \neq Y_i^{\ell-1}$ conditioning on $\R_{<i}$ and $e_i$ is given by
\begin{align}
\label{eq:Y-eq2}
\Pr{ Y^\ell_{i} \neq Y^{\ell-1}_{i} \mid \R_{<i},e_i } & = \abs{\mu^{\sigma_i^\ell}_{e_i} - \mu^{\sigma_i^{\ell-1}}_{e_i}}, 
\end{align}
where $\mu_e^S=\frac{\mu(S \cup \{e\})}{\mu(S \cup \{e\})+\mu(S \setminus \{e\})}$ is the marginal probability.

Let $\G\subseteq 2^E$ be the good event in \cref{cond:good-event-property}.
When $\sigma^{\ell}_i,\sigma^{\ell-1}_i\in\G$, by \cref{cond:good-event-property}, 
\begin{align}
\label{eq:temp}
\Pr{Y_i^\ell \neq Y_i^{\ell-1} \mid \R_{<i}} = \frac{1}{m} \sum_{e \in E} \abs{\mu^{\sigma_i^\ell}_{e} - \mu^{\sigma_i^{\ell-1}}_{e}} \le \frac{|\sigma_i^{\ell} \bigoplus \sigma_i^{\ell-1}|}{2m}.
\end{align}
Recall the construction of $\sigma_i^\ell$ in \Cref{line:sigma}.
Configurations $\sigma^\ell_i$ and $\sigma^{\ell-1}_i$ differ at site $e\in E$ if and only if $\pred_i(e)>1$ and $Y^{\ell-1}_{\pred_i(e)} \neq Y^{\ell-2}_{\pred_i(e)}$. Note that for different $e_1,e_2 \in E$ with $\pred_i(e_1),\pred_i(e_2)>0$, it holds that $\pred_i(e_1) \neq \pred_i(e_2)$. Therefore, $|\sigma^\ell_i \bigoplus \sigma^{\ell-1}_i |$ can be bounded by
\begin{align}
\label{eq:ham-distance}
\left|\sigma^\ell_i \bigoplus \sigma^{\ell-1}_i \right| \leq \sum_{j\in \Pred_i} I\left[Y^{\ell-1}_j \neq Y^{\ell-2}_j\right].
\end{align}
Combining~\eqref{eq:temp},~\eqref{eq:ham-distance}, and by the law of total probability, we have
\begin{align}
\label{eq:Y-eq4}
\Pr{ Y^\ell_{i} \neq Y^{\ell-1}_{i} } \leq
\frac{1}{2m} \sum_{j\in \Pred_i} \Pr{ Y^{\ell-1}_{j} \neq Y^{\ell-2}_{j} } + \Pr{ (\sigma^\ell_i \notin \G)\vee (\sigma^{\ell-1}_i \notin \G )}.
\end{align}
Now, fix any $1 \le t \le \tPA$, $1 \leq i \leq \tGD$.
By the definition of $\sigma_i^t$ and the updating rule of $Y_i^t$, 
\begin{align*}
    \forall e \in E, \quad I[e \in \sigma_i^t] = 
    \begin{cases}
        Y^{t-1}_{\pred_i(e)}, & t > 1 \land \pred_i(e) \geq 1, \\
        I[e \in \init], & \text{otherwise}. 
    \end{cases}
\end{align*}
It is straightforward to verify that, in \Cref{alg:BP-simulation}, $Y^{t-1}_{\pred_i(e)}$ and $I[e \in \init]$ are independently generated according to $\mu^{\sigma^{t-1}_{\pred_i(e)}}_e$ and $\mu^E_e$, respectively.
We can define the function $S: E \to 2^E$ as
\begin{align*}
    \forall e \in E, \quad S(e) := \begin{cases}
        \sigma^{t-1}_{\pred_i(e)}, & t > 1 \land \pred_i(e) \geq 1, \\
        E, & \text{otherwise}.
    \end{cases} 
\end{align*}
Note that $\sigma^t_i$ is generated by the same rule as \Cref{item:good-event-2} of \Cref{cond:good-event-property} using this $S$.
This, according to \Cref{item:good-event-2} of \Cref{cond:good-event-property}, implies that $\Pr{\sigma^t_i \not\in \+G} \leq \eta$.
Note that $i, t$ are fixed arbitrarily. By a union bound, we have
\begin{align}
\label{eq:Y-eq6}
\Pr{ (\sigma^\ell_i \notin \G)\vee (\sigma^{\ell-1}_i \notin \G )} \leq \Pr{\sigma^\ell_i \notin \G} + \Pr{\sigma^{\ell-1}_i \notin \G} \leq 2\eta.
\end{align}

Therefore,
~\eqref{eq:GD-induction-ineq} holds. This completes the proof of \cref{lemma:GD-dependency-path}.
\end{proof}

\cref{thm:dtv-xALG-GD} is a corollary of \cref{lemma:GD-dependency-path}.
\begin{proof}[Proof of \Cref{thm:dtv-xALG-GD}]
We couple $\xALG$ and $\xGD$ by using the same random choices  $\init$ and $\left(e_i,\randseed_i\right)_{i=1}^{\tGD}$ in the two processes. 
Let $(Z_i)_{i=0}^{\tGD}$ be the sequence  of configurations defined by 
\begin{align*}
Z_i := \left\{e \in E \mid (j_i(e) = 0 \land e\in \init) \lor (j_i(e)\geq 1 \land \GDResult^{\tPA}_{j_i(e)}=1)\right\},
\end{align*}
where $j_i(e) = \max\{j\ge 1\mid j\leq i \land e_j=e \}\cup\{0\}$.
Recall the rules for $Y_i$ being updated in \Cref{line:update-rules}, 
and the construction of $\sigma_i^\ell$ in \Cref{line:sigma}.  It holds that for all $ 1 \le i \le \tGD$,
\begin{align}
    Z_i = 
    \begin{cases}
    Z_{i-1} \cup \{e_i\} & \randseed_i < \frac{\mu(\sigma^{\tPA}_i \cup \{e_i\})}{\mu(\sigma^{\tPA}_i \cup \{e_i\}) + \mu(\sigma^{\tPA}_i \setminus \{e_i\})},\\
    Z_{i-1} \setminus \{e_i\} & otherwise.
    \end{cases}\label{eq:Z-process}
\end{align}
Now assume that $Y^{\tPA} = Y^{\tPA-1}$. For all $ 1 \le i \le \tGD$, $Z_{i-1}$ satisfies
\begin{align*}
Z_{i-1} &=  \left\{e \in E \mid (j_{i-1}(e) = 0 \land e\in \init) \lor (j_{i-1}(e)\geq 1 \land \GDResult^{\tPA}_{j_{i-1}(e)}=1)\right\}\\
&= \left\{e \in E \mid (\pred_i(e) = 0 \land e\in \init) \lor (\pred_i(e)\geq 1 \land \GDResult^{\tPA}_{\pred_i(e)}=1)\right\}\\ 
&= \left\{e \in E \mid (\pred_i(e) = 0 \land e\in \init) \lor (\pred_i(e)\geq 1 \land \GDResult^{\tPA-1}_{\pred_i(e)}=1)\right\}
= \sigma^{\tPA}_i.
\end{align*}
Therefore, the $\sigma^{\tPA}_i$ in \eqref{eq:Z-process} can be replaced with $Z_i$, and hence the sequence $(Z_i)_{i=0}^{\tGD}$ is identical to $(X_i)_{i=0}^{\tGD}$ because the two processes have the same transition rules  \eqref{eq:Z-process} and \eqref{eq:GD-process} and initial state $Z_0=X_0$. 
Hence,
\begin{align*}
    \Pr{ \xALG \neq \xGD } \leq \Pr{Y^{\tPA} \neq Y^{\tPA - 1}} \leq \left( 2^{-(\tPA-1)}  + 4 \eta  \right) \cdot \tGD,
\end{align*}
where the last inequality follows from \cref{lemma:GD-dependency-path}. This concludes the proof.
\end{proof}

\section{Low-temperature random clusters (Proof of \texorpdfstring{\Cref{thm:GD-correctness}}{Theorem 4.3})} \label{appendix:verify-cond-GD-sim}


Recall that $G = (V, E)$, $\*p, \*\lambda$, and $\epsilon$ are the input of \Cref{alg:BP-simulation}, where it holds that $n=\abs{V} \geq 3$, $\epsilon \leq 1/2$, and 
    \begin{align*}
        (1-p_{\min}) \log n \leq \min\left\{\-{e}^{-40} \exp\tp{-\frac{5\log (1/\epsilon)}{\log n}}, \frac{1 - \lambda_{\max}}{27} \right\}.
\end{align*}
Furthermore, $T^{\-{GD}} \ge \ctp{2m(\log m + \log(4/\epsilon))}$ and $\tPA \ge \lceil 3\log \left( 4\tGD/\epsilon \right) \rceil$ are parameters used in \cref{alg:BP-simulation}.
These assumptions will be used throughout.

First, under the above assumption, it is known that
the Glauber dynamics for the distribution $\mu = \mu^{\-RC}_{E, \*p, \*\lambda}$ of the random cluster model is rapidly mixing.

\begin{lemma}[\text{\cite[Theorem 5.1]{chen2023nearlinear}}] \label{lemma:GD-mixing}
Let $(X_i)_{i=0}^{\tGD}$ be the Glauber dynamics on the distribution $\mu$ with the initial state $X_0$ generated as in \cref{line:ini-X} of \cref{alg:BP-simulation},
    then it holds that $$\DTV{X_{T^{\-{GD}}}}{\mu} \leq \epsilon / 2.$$
\end{lemma}

\begin{remark}
In~\cite{chen2023nearlinear}, the initial state of the Glauber dynamics was set to $E$ instead of the random $X_0$. Though, the proof of~\cite[Theorem 5.1]{chen2023nearlinear} holds so long as the initial state follows the law of a product distribution whose marginal probabilities are bounded from below by $1-3(1-p_{\min})\log n$. By the argument in~\cite[Lemma 5.5]{chen2023nearlinear}, this is satisfied by the random initial state $X_0$ in \Cref{lemma:GD-mixing}. 
\end{remark}

We are now ready to prove \Cref{thm:GD-correctness}.
\begin{proof}[Proof of \Cref{thm:GD-correctness}]
We first prove \Cref{thm:GD-correctness} by assuming \Cref{cond:good-event-property} with parameter $\eta = \frac{\epsilon}{16 \tGD}$.
By \cref{thm:dtv-xALG-GD} and the assumption of parameters, 
\begin{align*}
    \DTV{\xALG}{X_{\tGD}} \le 2\tGD \cdot 2^{- 3\log (4\tGD/\epsilon)} + 4\tGD \cdot \frac{\epsilon}{16\tGD} \le \epsilon / 2.
\end{align*}
Therefore, by \Cref{lemma:GD-mixing} and a triangle inequality, it holds that
\begin{align*}
    \DTV{\xALG}{\mu} &\leq \DTV{\xALG}{\xGD} + \DTV{\xGD}{\mu} \leq \epsilon. 
\end{align*}
The remaining proof is dedicated to verifying \cref{cond:good-event-property} with parameter $\eta = \frac{\epsilon}{16\tGD}$.


\paragraph*{Verifying \Cref{item:good-event-1} of \Cref{cond:good-event-property}:} let $\+C$ be the family of vertex sets with large expansion:
\begin{align*}
    \+C := \left\{ S \subseteq V \mid  \abs{S} \leq n/2 \text{ and } \abs{E(S, V\setminus S)} \geq \abs{S} \log n \right\},
\end{align*}
where $E(S, V\setminus S)$ is the set of edges between $S$ and $V\setminus S$. 
We define the good event
\begin{align*}
    \+G := \left\{X\subseteq E \mid \forall S\in \+C, \abs{X \cap E(S, V\setminus S)} > 0 \right\}.
\end{align*}
We remark that this is the same good event used in \cite[Section 5]{chen2023nearlinear}.
In has been established in \cite[Eq.~(14)]{chen2023nearlinear} that
\begin{align*}
    \sum_{e \in E \setminus\{f\}} \abs{\mu^X_e - \mu^Y_e} = 0.
\end{align*}
Note that we also have $\abs{\mu^X_f - \mu^Y_f} = 0$ by definition. Therefore,
\begin{align} \label{eq:1-dist-decay}
    \sum_{e \in E} \abs{\mu^X_e - \mu^Y_e} \leq \frac{1}{2}.
\end{align}
For $X, Y \in \+G$ with $\abs{X \oplus Y} > 1$, let $e_1, e_2, \cdots, e_i$ be edges in $X \setminus Y$ and let $f_1, f_2, \cdots, f_j$ be edges in $Y \setminus X$.
Construct following configuration path $P := (P_0, P_1, \cdots, P_{i + j})$ from $X$ to $Y$:
\begin{itemize}
    \item let $P_0 = X$;
    \item for $1 \leq t \leq j$, let $P_t = P_{t-1} \cup \{f_t\}$;
    \item for $j < t$, let $P_t = P_{t-1} \setminus \{e_{t-j}\}$.
\end{itemize}
Note that $\abs{P_{t-1} \oplus P_{t}} = 1$ for all $1 \leq t \leq i + j$, and $P_t \in \+G$ for all $0 \leq t \leq i + j$.
Hence,
\begin{align*}
    \sum_{e\in E} \abs{\mu^X_e - \mu^Y_e}
    \leq \sum_{e \in E} \sum_{1 \leq t \leq i + j} \abs{\mu^{P_{t-1}}_e - \mu^{P_t}_e}
    = \sum_{1 \leq t \leq i + j} \sum_{e \in E} \abs{\mu^{P_{t-1}}_e - \mu^{P_t}_e}
    \leq \frac{\abs{X \oplus Y}}{2},
\end{align*}
where the last inequality follows from \eqref{eq:1-dist-decay} and the fact that $i + j = \abs{X \oplus Y}$.

\paragraph*{Verifying \Cref{item:good-event-2} of \Cref{cond:good-event-property}:}
Fix $S: E \to 2^E$, the random configuration $\sigma$ is generated by including each $e \in E$ independently with probability $\mu^{S(e)}_e$.

\begin{lemma}[\text{\cite[{Lemma 5.6}]{chen2023nearlinear}}] \label{lemma:marginal-bound}
    For any $T \subseteq E$ and $e \in E$, it holds that
    \begin{align*}
         \mu^T_e \geq 1 - 3(1-p_{\min})\log n.
    \end{align*}
\end{lemma}

For simplicity, let $K = (1 - p_{\min}) \log n$, and \Cref{lemma:marginal-bound} actually implies a marginal lower bound for $\sigma$, that is for each $e \in E$, $\Pr{e \in \sigma} \geq 1 - 3K$.

Since each edge $e$ is added to $\sigma$ independently, by the definition of $\+G$,
\begin{align*}
  \Pr{\sigma \not\in \+G}
  \le \sum_{S \in \+C} (3K)^{\abs{E(S,V \setminus S)}} 
  \le \sum_{S \in \+C} (3K)^{\abs{S} \log n} 
  \le \sum_{j=1}^{+\infty} n^j n^{j \log(3K)}
  \le n^{\log (27K)}.
\end{align*}
Recall that we choose 
\begin{align*}
    T^{\-{GD}} 
    &= \ctp{2m(\log m + \log(2/\epsilon))} \leq 4n^2 \log \tp{\frac{4n^2}{\epsilon}} \leq \frac{16n^4}{\epsilon} \leq \frac{n^{10}}{\epsilon},
\end{align*}
where in the last inequality, we use the assumption that $n \geq 3$.
By assumption, it holds that
\begin{align*}
    n^{\log(27K)} 
    &= \exp\tp{\log n \cdot \log (27K) } \\
    &\leq \exp\tp{\log n \cdot \log \tp{\-{e}^{-40} \exp\tp{- \frac{3\log(1/\epsilon)}{\log n}}}} \\
    &= \exp\tp{\log n \cdot \tp{-40 - \frac{3\log(1/\epsilon)}{\log n}}} \\
    &= \exp\tp{-40 \log n - 3\log(1/\epsilon)} = \frac{\epsilon^3}{n^{40}} \leq \frac{\epsilon}{16 T^{\-{GD}}},
\end{align*}
where in the last inequality, we use $n \geq 3$.
This finishes the verification. 
\end{proof}

\section{Analysis of field dynamics (Proof of \texorpdfstring{\Cref{thm:field-sim}}{Theorem 3.1})} \label{appendix:wrap-up}
We now prove \Cref{thm:field-sim}.
The \Cref{item:field-efficiency} (efficiency) of \Cref{thm:field-sim} will be proved in \Cref{sec:field-efficiency},
and \Cref{item:field-correctness} (accuracy)  of \Cref{thm:field-sim} will be proved in \Cref{sec:field-correctness}.

\subsection{Efficiency} \label{sec:field-efficiency}


We first bound the efficiency of \Cref{alg:BP-simulation}, assuming the setting of parameters in \cref{tabel:parameters}.
\begin{proposition}
\label{pro:BP-running-time}
\Cref{alg:BP-simulation} teriminates in $O(\tPA\cdot\log m)$ parallel time steps on $O(m\cdot \tGD)$ machines.
\end{proposition}

\begin{proof}
    The parallel complexity of \cref{alg:BP-simulation} is dominated by \cref{line:pred} and \cref{line:weight}. 
    
    \cref{line:pred} requires to compute $\pred_i(e)=\max\{j\ge 1\mid j<i \land e_j=e \}\cup\{0\}$ for each $1\leq i\leq \tGD$ and $e\in E$ in parallel. For each edge $e\in E$, the update list $\lis_e$ is defined as: 
    \begin{align}
        \lis_e := \{ j\geq 1 \mid e_j=e\}\cup\{0\}
        \label{eq:upd-list}
    \end{align}
    To implement \cref{line:pred}, for each $e\in E$, we first store a sorted update list $\lis_e$ using parallel merge sort \cite{cole1988parallel}, which costs $O(\log \tGD)$ parallel time on $O(\tGD)$ machines.
    Then each predecessor $\pred_i(e)$ can be computed in $O(\log \tGD)$ parallel time on $O(1)$ machines by a binary search on sorted $\lis_e$. The whole process costs $O(\log \tGD)$ parallel time on $O(m \cdot \tGD) 
    $ machines in total.
    
    
    In \cref{line:weight}, the marginal distribution $\mu_{e_i}^\sigma=\frac{\mu(\sigma\cup\{e_i\})}{\mu(\sigma\cup\{e_i\})+\mu(\sigma\setminus \{e_i\})}=\frac{w(\sigma\cup\{e_i\})}{w(\sigma\cup\{e_i\})+w(\sigma\setminus \{e_i\})}$ 
    can be computed by calculating the weights of $\sigma\cup\{e_i\}$ and $\sigma\setminus\{e_i\}$ (defined in \eqref{eq:RC-weight}), whereas each weight can be computed by a connected component algorithm such as \cite{shiloach1980log}, which costs $O(\log m)$ parallel time on $O(n+m)$ machines. 
    Thus, each iteration of \cref{line:iter} costs $O(\log m)$ parallel time on $O(m \cdot \tGD) 
    $ machines, and these $O(m \cdot \tGD)$ machines can be re-used in the next iteration. 
    In \cref{line:ini-X}, the marginal distribution $\mu_e^E$ can be computed in the same method, which costs $O(\log m)$ parallel time on $O(n+m)$ machines for each edge $e\in E$.
    
    For other costs, the preprocessing part (from \cref{line:ini1} to \cref{line:ini2}) of \cref{alg:BP-simulation} can be computed in $O(1)$ parallel time on $O(\tGD)
    $ machines; and in \cref{line:res}, the result $X$ of \cref{alg:BP-simulation} can be computed in $O(1)$ parallel time on $O(m)$ machines, since $j_e$ is maximum of the sorted list $\lis_e$.

    Overall, 
    \Cref{alg:BP-simulation} runs in $O(\tPA \cdot \log m)$ parallel time on $O(m \cdot \tGD)$ machines.
\end{proof}

\begin{proof}[Proof of \Cref{item:field-efficiency} of \Cref{thm:field-sim}]
    When $\abs{V} \leq N_0$, \Cref{alg:field} generates $X$ by brute force, which costs $2^{O(N_0^2)} = \exp\tp{O_\delta\tp{\frac{\log(1/\epsilon)}{\log n}}}$  in total work.
    When $\abs{V} > N_0$, \Cref{alg:field} terminates within $\tFD$ iterations, and in each iteration, it generates a random $S$ using $O(1)$ parallel time on $O(m)$ machines and calls \Cref{alg:BP-simulation}.
    By \Cref{pro:BP-running-time}, we have that \Cref{alg:field} terminates within
        $O(\tPA\cdot \log m) \cdot \tFD = \tp{\exp\tp{\frac{\log(1/\epsilon)}{\log n}} \cdot \log n}^{O_{\delta}(1)}$
    parallel time steps on $O(m \cdot \tGD) = O_{\delta}(m^2\log(n/\epsilon))$ machines.
\end{proof}


\subsection{Accuracy of sampling} \label{sec:field-correctness}
We now prove \Cref{item:field-correctness} of \Cref{thm:field-sim}, the accuracy of sampling of \Cref{alg:field}, still assuming the setting of parameters in \cref{tabel:parameters}.
A key step has already been provided in \Cref{thm:GD-correctness}.
To complete the proof, we are going to establish the followings.

\begin{lemma}\label{lemma:correctness}
   %
   %
   If for every $U \subseteq E$, the output  $Y$ of $\ParGD{(V, U)}{\*p^\star}{\*\lambda}{(2\tFD)^{-1}\epsilon}$ always satisfies $\DTV{Y}{\mu^{\-{RC}}_{U, \*p^\star, \*\lambda}} \leq (2\tFD)^{-1}\epsilon$, 
   then the output $X$ of \Cref{alg:field} satisfies $\DTV{X}{\mu^{\-{RC}}_{E, \*p, \*\lambda}} \leq \epsilon$.
\end{lemma}

\begin{lemma} \label{lemma:GD-sim-correctness}
    Assuming $n=|V|\ge N_0$,
    it holds that for every $U \subseteq E$, the output $Y$ of the subroutine $\ParGD{(V, U)}{\*p^\star}{ \*\lambda}{(2T^{\-{FD}})^{-1}\epsilon}$ satisfies 
    \begin{align*}
        \DTV{Y}{\mu^{\-{RC}}_{U, \*p^\star, \*\lambda}} \leq \frac{\epsilon}{2 T^{\-{FD}}}.
    \end{align*}
\end{lemma}

\Cref{item:field-correctness} of \Cref{thm:field-sim} follows directly by combining \Cref{lemma:correctness} and \Cref{lemma:GD-sim-correctness}. 

\Cref{lemma:correctness} follows from the rapid mixing of the field dynamics.
\begin{lemma}[\text{\cite[Lemma 3.5]{chen2023nearlinear}}]\label{lemma:mixing-field}
  The mixing time of the field dynamics satisfies
  \begin{align*}
    \forall \epsilon \in (0,1), \quad T_{\mathrm{mix}}(\epsilon) \leq \tp{\frac{\mathrm{e}}{\theta}}^{5(1-\lambda_{\max})^{-2}}\tp{2 \log n + \log \log \frac{2}{p_{\min}} + \log \frac{1}{2\epsilon^2}}.
  \end{align*}
\end{lemma}

\begin{proof}[Proof of \Cref{lemma:correctness}]
  Without loss of generality, we assume that $\abs{V} > N_0$.
   For $0 \le t \le T^{\mathrm{FD}}$, let $X_t$ be the configuration $X$ after the $t$-th iteration of \Cref{alg:field}, and let $Y_t$ be the 
   configuration $Y$ after the $t$-th iteration of the field dynamics $P^{\mathrm{FD}}_\theta$ starting from the same initial state $X_0=Y_0=E$.
   
   Let
   $\mu=\mu^{\-{RC}}_{E, \*p, \*\lambda}$ for short.
   By triangle inequality and the coupling lemma, it holds that
   \[
     \DTV{X_{T^{\-{FD}}}}{\mu}
     \leq \DTV{X_{T^{\-{FD}}}}{Y_{T^{\-{FD}}}} + \DTV{Y_{T^{\-{FD}}}}{\mu}
     \leq \Pr{X_{T^{\-{FD}}} \neq Y_{T^{\-{FD}}}} + \DTV{Y_{T^{\-{FD}}}}{\mu},
   \]
  for any coupling $(X_t,Y_t)$ of the two processes $X_t$ and $Y_t$.
  And it is obvious that if $X_{T^{\-{FD}}} \neq Y_{T^{\-{FD}}}$, then there must exist $1 \leq i \leq T^{\-{FD}}$ such that $X_i \neq Y_i$ but $X_j = Y_j$ for all $0\le j < i$.
  Hence, it holds that
  \begin{align*}
     \Pr{X_{T^{\-{FD}}} \neq Y_{T^{\-{FD}}}}
     &\le \sum_{i=1}^{T^{\-{FD}}} \Pr{X_i \neq Y_i \text{ and } \forall j < i, X_j = Y_j} 
     \leq \sum_{i=1}^{T^{\-{FD}}} \Pr{X_i \neq Y_i \mid X_{i-1} = Y_{i-1}}.
   \end{align*}
    For any $1 \le i \le T^{\-{FD}}$, conditioned on $X_{i-1} = Y_{i-1}$, we construct a coupling of $X_i$ and $Y_i$:
    \begin{enumerate}
      \item generate a random $S\subseteq E$ by adding each $e \in E$ independently with probability $\theta$;
      \item generate $(X_i, Y_i)$ according to the optimal coupling of $\widehat{\mu}^{\-{RC}}_{U,\*p^\star,\*\lambda}$ and $\mu^{\-{RC}}_{U,\*p^\star,\*\lambda}$, where $U = S \cup X_{i-1}$ and $\widehat{\mu}^{\-{RC}}_{U,\*p^\star,\*\lambda}$ is the distribution of the output of $\ParGD{(V, U)} {\*p^\star}{\*\lambda}{(2T^{\-{FD}})^{-1}\epsilon}$.
    \end{enumerate}
  By the coupling lemma and the assumption of \Cref{lemma:correctness}, it holds that
  \begin{align*}
     \Pr{X_i \neq Y_i \mid X_{i-1} = Y_{i-1}} &\leq \max_{U \subseteq E} \DTV{\widehat{\mu}^{\-{RC}}_{U,\*p^\star,\*\lambda}}{\mu^{\-{RC}}_{U, \*p^\star, \*\lambda}} \le \frac{\epsilon}{2T^{\mathrm{FD}}}.
   \end{align*}
   By \Cref{lemma:mixing-field} and our choice of $T^{\-{FD}}$, we have
   \begin{align*}
   \DTV{X_{T^{\-{FD}}}}{\mu} &\leq T^{\-{FD}} \frac{\epsilon}{2 T^{\-{FD}}} + \frac{\epsilon}{2} = \epsilon. 
   \end{align*}
\end{proof}

Now, it only remains to prove \Cref{lemma:GD-sim-correctness}.

\begin{proof}[Proof of \Cref{lemma:GD-sim-correctness}]
    It is sufficient to verify the condition of \Cref{thm:GD-correctness} with error bound $\frac{\epsilon}{2\tFD}$ on the new instance $\mu^{\-{RC}}_{U,\*p^\star, \*\lambda}$:
    \begin{align} \label{eq:cond-gd-mix}
        (1 - p^\star_{\min}) \log n \leq \min\left\{\-{e}^{-40} \exp\tp{-\frac{5\log (2 T^{\-{FD}}/\epsilon)}{\log n}}, \frac{1 - \lambda_{\max}}{27} \right\}.
    \end{align}
    By our assumption $n \ge N_0 \ge \max \left\{4,\frac{3}{p_{\min}}, {\color{black} \sqrt{\frac{\log \tp{2/\epsilon^2}}{\log n}} } \right\}$, it holds that
    \begin{align} 
        \log T^{\mathrm{FD}} 
        %
        \label{eq:TFD-ub}
        \le 6(1-\lambda_{\max})^{-2} \log \tp{1/K} + 14(1-\lambda_{\max})^{-2} \log n,
    \end{align}
    where for convenience, we denote 
    \begin{align*}
        K := \theta \cdot \frac{\log n}{p_{\min}} = \e^{-100} \exp\tp{\frac{10\log(\epsilon/2)}{\log n}} \exp\tp{-140(1 - \lambda_{\max})^{-2}}.
    \end{align*}
    %
    The first term in the right hand side of \eqref{eq:cond-gd-mix} can be bounded by
    \begin{align*}
        \e^{-40} \exp&\tp{-\frac{5\log (2 T^{\-{FD}}/\epsilon)}{\log n}} 
        = \e^{-40} \exp\tp{-\frac{5\log T^{\-{FD}}}{\log n} + \frac{5\log(\epsilon/2)}{\log n}}  \\
        %
        %
        (\text{by \eqref{eq:TFD-ub}}) \quad
        &\geq \e^{-50} \exp\tp{\frac{5\log(\epsilon/2)}{\log n}} \exp\tp{-70 (1 - \lambda_{\max})^{-2} +\frac{30(1 - \lambda_{\max})^{-2}\log K}{\log n}} \\
        &\geq \e^{-50} \exp\tp{\frac{5\log(\epsilon/2)}{\log n}} \exp\tp{-70 (1 - \lambda_{\max})^{-2}} K^{1/2},
    \end{align*}
    where the last inequality holds since $n \geq N_0 \geq \exp\tp{60 (1 - \lambda_{\max})^{-2}}$.
    
    Note that the $\*p^\star= \frac{\*p}{\*p+\theta (1-\*p)}$ constructed in \Cref{alg:field} satisfies
    \begin{align*}
        p^{\star}_{\min} = \min_{e \in E} p^{\star}_e = \frac{p_{\min}}{p_{\min}+\theta (1-p_{\min})} \ge \frac{\log n}{\log n+K} \ge 1- \frac{K}{\log n}.
    \end{align*}
    Therefore, we have the following two bounds on $(1 - p^\star_{\min}) \log n$:
    \begin{align}
        \nonumber
        \tp{1-p^{\star}_{\min}} \log n \le K &\le \e^{-50}\exp\tp{\frac{5\log(\epsilon/2)}{\log n}}\exp\tp{-70(1-\lambda_{\max})^{-2}} K^{1/2} \\
        \label{eq:p1}
        &\le  \e^{-40} \exp\tp{-\frac{5\log (2 T^{\-{FD}}/\epsilon)}{\log n}}, \\
        \label{eq:p2}
        \tp{1-p^{\star}_{\min}} \log n \le K &\le \exp \tp{-140(1-\lambda_{\max})^{-2}} \le \frac{(1-\lambda_{\max})^2}{140} \le \frac{1-\lambda_{\max}}{27}.
    \end{align}
    Combining ~\eqref{eq:p1} and ~\eqref{eq:p2}, we conclude the proof of \Cref{lemma:GD-sim-correctness}.
\end{proof}

\section{Analysis of restricted Gaussian dynamics}
\label{appendix:gaussian-dynamics}
In this section, we verify \cref{item:gaussian-1}, \cref{item:gaussian-2}, and \cref{item:gaussian-3} stated in \cref{sec:gaussian-dynamics}, completing the overall proof.

\cref{item:gaussian-1} requires an efficient method for approximately sampling in parallel from the Gaussian distribution $\mathcal{N}(x_{i-1},J^{-1})$.
The Langevin process is a well-established technique for Gaussian sampling, and recent work has introduced a parallelization for this approach.

\begin{theorem}[\text{\cite[Corollary 14]{anari2024fast}}] \label{thm:gaussian-sampler}
    Let $\alpha, \beta > 0$ be  constants. Define $\kappa := \beta/\alpha$.
    Let $\pi = \exp(-V)$ be a distribution over $\mathbb{R}^d$ such that $V:\mathbb{R}^d \to \mathbb{R}$ satisfies 
    $$0 \prec \alpha I \preceq \nabla^2 V(x) \preceq \beta I$$ 
    for all $x \in \mathbb{R}^d$.
    Given the minimizer $x^\star$ of $V$ and an error bound $\epsilon_0$, there exists a parallel sampling algorithm $\+A$ that generates  $Y$ such that $\DTV{Y}{\pi} \leq \epsilon_0$ by using $O(\kappa\log \kappa \log^2(d/\epsilon_0^2))$ parallel rounds and at most $7\max\left\{\frac{\kappa d}{\epsilon_0^2}, \kappa^2\right\}$-gradient evaluations per round. 
\end{theorem}

Simulating \cref{line:gaussian-noising} requires sampling from $\mathcal{N}(0,J^{-1})$. 
By \cref{thm:gaussian-sampler}, we define the potential function $V(x) = \frac{1}{2}x^\top Jx$. 
Under the constraint $0 \prec \frac{\eta}{2} I \preceq J \preceq (1-\frac{\eta}{2}) I$, we set $\kappa = 2/\eta$ and $d = n$.
Evaluating the gradient $\nabla V(x)$ takes $O(\log n)$ parallel rounds on $O(n^2)$ machines.
Thus, the total complexity of the parallelized Langevin algorithm is $O(\kappa\log \kappa \log^2(d/\epsilon_0^2)))\cdot O(\log n)=O_\eta(\log^3\tp{n/\epsilon})$ and each parallel round requires
$7\max\left\{\frac{\kappa d}{\epsilon_0^2}, \kappa^2\right\}\cdot O(n^2) = \widetilde{O}_\eta(n^3/\epsilon_0^2)$ machines.
Thus, \cref{item:gaussian-1} is verified.

Verifying \cref{item:gaussian-2} follows directly by expanding the distribution formula in \cref{line:gaussian-denoising}.
Since the interaction matrix $J$ is symmetric, we have
\begin{align*}
    \mu^{\-{Ising}}_{J, h}(x) \cdot \exp\left(-\frac{1}{2}(y_i-x)^\top J(y_i-x)\right) &=  \exp\left(\frac{1}{2} x^\top Jx + h^\intercal \sigma -\frac{1}{2}(y_i-x)^\top J(y_i-x)\right)
    \\ &=\exp\left(\inner{y_i^\top J+h^\top}{x}- \frac{1}{2}y_i^\top J y_i\right)
\end{align*}
It follows that each dimension in the above distribution is independent.
The cost of computing $y_i^\top J+h^\top$ is $O(\log n)$ parallel rounds on $O(n^2)$ machines.


Verifying \cref{item:gaussian-3} follows from the rapid mixing of the restricted Gaussian dynamics, which has been established recently.

\begin{theorem}[\cite{chen2022localization,chen2024rapid}] \label{thm:proximal-sampler-noising}
    Let $\delta \in (0,1)$ be a constant, and let $P$ be the transition matrix of the restricted Gaussian dynamics.
    If $\norm{J}_2 \leq 1 - \delta$, then the restricted Gaussian dynamics exhibits entropy decay with rate $\delta$, i.e., for all $f: \Omega(\mu)\to \mathbb{R}_{\geq 0}$,
    \begin{align*}
        \Ent[\mu]{P f} \leq (1 - \delta) \Ent[\mu]{f}.
    \end{align*}
Consequently, by Pinsker’s inequality, we obtain
\begin{align*}
        \DTV{X}{\mu} \leq \sqrt{\DKL{X}{\mu}/2} \leq \sqrt{\exp\tp{-\delta T^{\mathrm{RGD}}} \log (1/\mu(x_0))},
    \end{align*}
    where $X$ is the output of \Cref{alg:gaussian} and $x_0$ is the initial state.
\end{theorem}

Since $0 \prec \frac{\eta}{2} I \preceq J \preceq (1-\frac{\eta}{2}) I$, it follows that for all $x, y \in \{\pm 1\}^n$, 
$$\frac{\mu(x)}{\mu(y)} \leq \exp\left(\frac{2}{\eta}\right).$$
Thus, for any $x_0 \in \Omega(\mu)$,  we have $1/\mu(x_0) \leq 2^n \cdot \exp(2/\eta)$.
Applying \Cref{thm:proximal-sampler-noising}, we ensure $\DTV{X}{\mu} \leq \epsilon/2$ by choosing
\begin{align*}
    T^{\-{RGD}} = \frac{2}{\eta} \tp{\log\frac{4}{\epsilon^2} + \log\tp{n\log 2 + \frac{2}{\eta}}} = O_\eta\tp{\log (n/\epsilon)}.
\end{align*}
Thus, \cref{item:gaussian-3} is verified.

\end{document}